\documentclass[11pt,english]{article}
\usepackage{color}
\usepackage[T1]{fontenc}
\usepackage[latin9]{inputenc}
\usepackage[a4paper]{geometry}
\usepackage{textcomp}
\usepackage{amsthm}
\usepackage{amsmath}
\usepackage{amssymb}
\usepackage{esint}
\usepackage{bbm}
\usepackage{hyperref}
\usepackage{babel}
\usepackage{amscd}
\usepackage{amsfonts}
\usepackage{amsthm}
\usepackage{mathrsfs}
\usepackage{dsfont}
\usepackage{mathtools}
\usepackage{enumerate}
\usepackage{bm}
\usepackage{bbm}
\usepackage{authblk}

\usepackage[english=american]{csquotes}

\usepackage{cite}

\numberwithin{equation}{section}

\DeclareMathOperator{\Tr}{Tr}
\DeclareMathOperator{\OpW}{Op_\hbar^w}
\DeclareMathOperator{\supp}{supp}
\DeclareMathOperator{\dist}{dist}
\DeclareMathOperator{\Real}{\mathrm{Re}}

\newcommand{\scp}[2]{\langle #1 , #2 \rangle}
\newcommand{\norm}[1]{\lVert #1 \rVert}
\newcommand{\abs}[1]{\left| #1 \right|}

\newcommand{\R}{\mathbb{R}}

\newcommand{\N}{\mathbb{N}}

\newtheorem{thm}{Theorem}[section]
\newtheorem{lemma}[thm]{Lemma}
\newtheorem{prop}[thm]{Proposition}
\newtheorem{corollary}[thm] {Corollary}

\theoremstyle{definition}

\newtheorem{assumption}[thm]{Assumption}

\theoremstyle{remark}
\newtheorem{remark}[thm]{Remark}

\begin{document}

\title{An optimal semiclassical bound on certain commutators}

\author{S{\o}ren Fournais and S{\o}ren Mikkelsen}

\affil{\small{Department of Mathematics, Aarhus University\\ Ny
    Munkegade 118\\ DK-8000 Aarhus C\\ Denmark}}

\date{\today}

\maketitle

\begin{abstract}
  We prove an optimal semiclassical bound on the trace norm of the
  following commutators $[\boldsymbol{1}_{(-\infty,0]}(H_\hbar),x]$,
  $[\boldsymbol{1}_{(-\infty,0]}(H_\hbar),-i\hbar\nabla]$ and
  $[\boldsymbol{1}_{(-\infty,0]}(H_\hbar),e^{itx}]$, where $H_\hbar$
  is a Schr\"odinger operator with a semiclassical parameter $\hbar$,
  $x$ is the position operator and $-i\hbar\nabla$ is the momentum
  operator. These bounds corresponds to a mean-field version of bounds
  introduced as an assumption by N. Benedikter, M. Porta and B. Schlein in a
  study of the mean-field evolution of a fermionic system.
\end{abstract}


\section{Introduction and main result}

We consider a Schr\"odinger operator $H_\hbar=-\hbar^2\Delta + V$
acting in $L^2(\R^d)$ with $d\geq2$. Here $\Delta$ is the Laplacian
acting in $L^2(\R^d)$ and $V$ is a real valued function. We will be
interested in the following trace norms of commutators:
\begin{align*}
  \norm{[\boldsymbol{1}_{(-\infty,0]}(H_\hbar),x_j]}_{1}, &&
  \norm{[\boldsymbol{1}_{(-\infty,0]}(H_\hbar),Q_j]}_{1} && \text{and}
  && \norm{[\boldsymbol{1}_{(-\infty,0]}(H_\hbar),e^{itx}]}_{1},
\end{align*}
where $Q_j=-i\hbar\partial_{x_j}$ and $x_j$ is the position operator
for $j \in\{1,\dots,d\}$. Moreover $\boldsymbol{1}_A$ denotes the
characteristic function of a set $A$ and $\norm{\cdot}_1$ denotes the
trace norm. The main theorem will be the bound for the first two
commutators and the bound on the last will follow as a corollary.

Let us specify the assumptions on the function $V$ for which we study
the operator $H_\hbar$.
\begin{assumption}\label{A.general_assumptions_on_V}
  Let $V:\R^d\rightarrow\R$ be a function for which there exists an
  open set $\Omega_V\subset\R^d$ and $\varepsilon>0$ such that
  \begin{enumerate}[$1)$]
  \item $V$ is in $C^\infty(\Omega_V)$.
	
  \item There exists an open bounded set $\Omega_\varepsilon$ such
    that $\overline{\Omega}_\varepsilon \subset \Omega_V$ such that
    $V\geq\varepsilon$ for all $x\in\Omega_\varepsilon^c$.

  \item $V\boldsymbol{1}_{\Omega_V^c}$ is an element of
    $L^1_{loc}(\R^d)$.
  \end{enumerate}
\end{assumption}
The assumption of smoothness in the set $\Omega_V$ is needed in order
to use the theory of pseudo-differential operators. The second
assumption is needed to ensure that we have non continuous spectrum in
$(-\infty,0]$ and enable us to localise the operator. The last
assumption is just to ensure that we can define the operator $H_\hbar$
by a Friedrichs extension of the associated form. We can now state our
main theorem:
\begin{thm}\label{A.Main_Theorem}
  Let $H_\hbar=-\hbar^2 \Delta + V$ be a Schr\"odinger operator acting
  in $L^2(\R^d)$ with $d\geq2$, where $V$ satisfies
  Assumption~\ref{A.general_assumptions_on_V} and let
  $Q_j=-i\hbar\partial_{x_j}$ for $j\in\{1,\dots,d\}$ futhermore, let $\hbar_0$ be a strictly positive number. Then the
  following bounds hold
  \begin{equation}\label{A.Bounds_to_prove}
    \norm{[\boldsymbol{1}_{(-\infty,0]}(H_\hbar),x_j]}_{1} \leq C \hbar^{1-d}
    \quad\text{and}\quad 
    \norm{[\boldsymbol{1}_{(-\infty,0]}(H_\hbar),Q_j]}_{1} \leq C \hbar^{1-d} ,
  \end{equation}
  for all $\hbar$ in $(0,\hbar_0]$, where $C$ is a positive constant.
\end{thm}
From Theorem~\ref{A.Main_Theorem} we get the corollary:
\begin{corollary}\label{A.cor_Main_Theorem}
  Let $H_\hbar=-\hbar^2 \Delta + V$ be a Schr\"odinger operator acting
  in $L^2(\R^d)$ with $d\geq2$, where $V$ satisfies
  Assumption~\ref{A.general_assumptions_on_V} futhermore, let $\hbar_0$ be a strictly positive number. Then the following
  bound holds
  \begin{equation}\label{A.Bound_to_prove_2}
    \norm{[\boldsymbol{1}_{(-\infty,0]}(H_\hbar),e^{i \langle t,x\rangle}]}_{1} \leq C \abs{t} \hbar^{1-d},
  \end{equation}
  for all $t$ in $\R^d$ and all $\hbar$ in $(0,\hbar_0]$, where $\langle t,x\rangle$ is the Euclidean
  inner product and $C$ is a positive constant.
\end{corollary}
Theorem~\ref{A.Main_Theorem} and
Corollary~\ref{A.cor_Main_Theorem} are semiclassical in the sense that they are of most interest in the cases where the semiclassical parameter $\hbar$ is small. The upper bound $\hbar_0$ on the semiclassical parameter is needed in order to control the constants as we do not have uniformity for $\hbar$ tending to infinity.   

The proofs of Theorem~\ref{A.Main_Theorem} and
Corollary~\ref{A.cor_Main_Theorem} are given in
section~\ref{A.Proof_of_Main}.  The proof of
Theorem~\ref{A.Main_Theorem} is divided into three parts. First a
local version of the theorem (see Theorem~\ref{A.Main_Local_Theorem
  _wnc}) is proven with a noncritical assumption
\eqref{A.non_critical_assumption}. This proof is based on local
Weyl-asymptotics proven in the paper \cite{MR1343781} and an $\hbar$
dependent dyadic decomposition which will be introduced in
the proof. In the first part we will not be considering the operator
$H_\hbar$ directly but an abstract operator $\mathcal{H}$ which satisfies
Assumption~\ref{A.local_assumptions} below. The abstract version is needed for the later multiscale argument.
 
The second part is to remove the non-critical condition by a
multiscale argument as in \cite{MR1343781} (see also
\cite{MR1631419,MR1240575}). The main idea is to make a partition of
unity and on each partition scale the operator in such a way that a
non-critical assumption is achieved and then use the theorem with the
non-critical condition. The final step in this part is to remove the
dependence of the partition by integration.

The third part is to first note that the theorem obtained in the
second part gives the desired estimate in the classically allowed
region $\{V<\varepsilon\}$ and then prove that the classically
forbidden region $\{V>\varepsilon\}$ contributes less to the error
term than the desired estimate. This is done by applying an Agmon type
bound on the eigenfunctions of the operator $H_\hbar$.

Commutator bounds of the type considered in this paper were introduced as assumptions in a series of papers by
N. Benedikter, M. Porta and B. Schlein
et.~al. \cite{MR3570479,MR3248060,MR3202863,MR3381147} where they
considered mean-field dynamics of fermions in different settings. The
bounds considered here are a first step to verifying their assumption,
since the bounds proven here correspond to a mean field version of the
bounds they need. The assumption reappeared in the paper
\cite{MR4009687}.

Already the mean-field version of the bounds, treated in this paper,
is non-trivial as they are optimal in terms of the semiclassical
parameter $\hbar$, which is easily seen by the calculus of
pseudo-differential operators. 


\section{Preliminaries}

\subsection{Assumptions and notation}

First we will describe the operators we are working with. Under
Assumption~\ref{A.general_assumptions_on_V} we can define the operator
$H_\hbar=-\hbar^2\Delta + V$ as the Friedrichs extension of the
quadratic form given by
\begin{equation*}
  \mathfrak{h}[f,g] = \int_{\R^d} \hbar^2\sum_{i=1}^d \partial_{x_i}f(x) \overline{\partial_{x_i}g(x)}  + V(x)f(x)\overline{g(x)}\;dx, \qquad f,g \in \mathcal{D}(\mathfrak{h}),
\end{equation*}
where
\begin{equation*}
  \mathcal{D}(\mathfrak{h}) = \left\{ f\in L^2(\R^d) | \int_{\R^d} \abs{p}^2 \abs{\hat{f}(p)}^2 \;dp<\infty \text{ and } \int_{\R^d} \abs{V(x)}\abs{f(x)}^2 \;dx <\infty \right\}.
\end{equation*}
In this set up the Friedrichs extension will be unique and
self-adjoint see e.g. \cite{MR0493420}. Moreover, we will also consider
operators that satisfy the following assumption
\begin{assumption}\label{A.local_assumptions}
  Let $\mathcal{H}$ be an operator acting in $L^2(\R^d)$ such that
  \begin{enumerate}[$1)$]
  \item $\mathcal{H}$ is selfadjoint and lower semibounded.

  \item There exists an open set $\Omega \subset \R^d$ and a
    realvalued function $V_{loc}$ in $C_0^ \infty(\R^d)$ such that
    $C_0^\infty(\Omega) \subset \mathcal{D}(\mathcal{H})$ and
    \begin{equation*}
      \mathcal{H} u = H_\hbar^{loc} u
    \end{equation*}
    for all $u$ in $C_0^\infty(\Omega)$, where
    $H_\hbar^{loc}=-\hbar^2 \Delta + V_{loc}$.
  \end{enumerate}
\end{assumption}
The above assumption is exactly the same as in \cite{MR1343781}. It is
important to note that the assumptions made on the the operator
$H_\hbar$ in Theorem~\ref{A.Main_Theorem} imply that $H_\hbar$
satisfies Assumption~\ref{A.local_assumptions} for a suitable
$V_{loc}$. When referring to this assumption further on we will omit
the $loc$ on the operator $H_\hbar^{loc}$ and the function $V_{loc}$
when we only consider an operator satisfying the assumption.

The construction of the operator via a Friedrichs extension will also
work for the local Schr\"odinger operator, where $V_{loc}$ is
$C_0^\infty(\R)$. But in this case the operator can also be
constructed as the closure of an $\hbar$-pseudo-differential operator
($\hbar$-$\Psi$DO) defined on the Schwarz space.  By an
$\hbar$-$\Psi$DO, $A = \OpW(a)$ we mean the operator with Weyl symbol
$a$, that is
\begin{equation*}
  \OpW(a)\psi(x) = \frac{1}{(2\pi \hbar)^{d}} \int_{\R^d} \int_{\R^d} e^{i \hbar^{-1} \scp{x-y}{p} } 
  a\left( \tfrac{x+y}{2},p\right) \psi(y) \;d{y}\;d{p},
\end{equation*}
for $\psi \in \mathcal{S}(\R^d)$ (the Schwarz space). The symbol $a$
is assumed to be in $C^\infty(\R_x^d \times \R_p^d)$ and to satisfy
the condition
\begin{equation}\label{A.sym_bound}
  \abs{\partial_x^\alpha \partial_p^\beta a(x,p)} \leq C_{\alpha,\beta} m (x,p),
\end{equation}
for all multi-index $\alpha$ and $\beta$ and some tempered weight
function $m$. The above integrals should be understod as oscillating
integrals. We need this as the results on Weyl-asymptotics 
needed is based on ($\hbar$-$\Psi$DOs). For more details see e.g. the
monographs \cite{MR897108,MR1735654,MR2952218}.

We call a number $E$ in $\R$ a non-critical value for a symbol $a$ if
\begin{equation*}
  (\nabla_x a(x,p),\nabla_p a(x,p)) \neq 0 \qquad \forall (x,p) \in a^{-1}(\{E\}). 
\end{equation*}
In the case where $a(x,p)=p^2+V(x)$ the non-critical condition can be
expressed only in terms of the function $V$ by assuming that
\begin{equation*}
  \abs{\nabla_x V(x)}^2 + \abs{E-V(x)}> 0, \qquad \forall (x,p) \in a^{-1}(\{E\}), 
\end{equation*}
since it is immediate that
\begin{equation*}
  \abs{\nabla_x a(x,p)}^2+\abs{\nabla_p a(x,p)}^2 = \abs{\nabla_x V(x)}^2 + 4\abs{E-V(x)}, \qquad \forall (x,p) \in a^{-1}(\{E\}).
\end{equation*}
%


\subsection{Optimal Weyl-asymptotics}

We are interested in optimal Weyl-asymptotics for an operator
$\mathcal{H}$ acting in $L^2(\R^d)$ satisfying
Assumption~\ref{A.local_assumptions}. When we only have one operator
we will not write the $loc$ subscript on the operator. In the
following we will denote the open ball with radius $R$ by
$B(0,R)$. For this kind of operators we have from \cite[Theorem 4.1]{MR1343781} the
following theorem:
\begin{thm}\label{A.Sobolev_4.1}
  Suppose the operator $\mathcal{H}$ acting in $L^2(\R^d)$ with
  $d\geq2$ obeys Assumption~\ref{A.local_assumptions} with
  $\Omega=B(0,4R)$ for $R>0$ and
  \begin{equation}\label{A.non_critical_assumption}
    \abs{V(x)} + \abs{\nabla V(x)}^2 + \hbar \geq c,
  \end{equation}
  for all $x$ in $B(0,2R)$ furthermore, let $\hbar_0$ be a strictly positive number. For $\varphi$ in $C_0^\infty(B(0,R/2))$ it
  holds that
  \begin{equation*}
    \Big|\Tr[\boldsymbol{1}_{(-\infty,0]}(\mathcal{H}) \varphi] -  \frac{1}{(2\pi\hbar)^d} 
    \int_{\R^d} \int_{\R^d} \boldsymbol{1}_{\{p^2 + V(x) \leq 0\}} (x,p) \varphi(x) \;d{x}\;d{p} \Big| \leq C\hbar^{1-d},
  \end{equation*}
  for $C$ a positive constant and all $\hbar$ in $(0,\hbar_0]$. The
  constant $C$ depends on the numbers $R$, $\hbar_0$ and $c$ in  \eqref{A.non_critical_assumption} and on the bounds on the derivatives of $V$ and $\varphi$.
\end{thm}
One can note that in our \enquote{non-critical}
assumption~\eqref{A.non_critical_assumption} in the above theorem
there has appeared an $\hbar$. This assumption would either imply that
$ \abs{V(x)} + \abs{\nabla V(x)}^2 \geq c/2$ or $\hbar\geq c/2$. In
the first case the assumption gives us our noncritical assumption. In
the second both sides will be finite and the formula can be made true
by an appropriate choice of constants.
\begin{prop}\label{A.Weyl_used_local}
  Suppose the operator $\mathcal{H}$ acting in $L^2(\R^d)$ with
  $d\geq2$ obeys Assumption~\ref{A.local_assumptions} with
  $\Omega=B(0,4R)$ for $R>0$. Moreover suppose there is an
  $\varepsilon>0$ such that
  \begin{equation}\label{A.non_critical_assumption_prop}
    \abs{V(x)-E} + \abs{\nabla V(x)}^2 + \hbar \geq c,
  \end{equation}
  for all $x$ in $B(0,2R)$ and all $E$ in
  $[-2\varepsilon,2\varepsilon]$ furthermore let $\hbar_0$ be a strictly positive number. For $\varphi$ in
  $C_0^\infty(B(0,R/2))$ and two numbers $a$ and $b$  such that
  \begin{equation*}
    -\varepsilon < a \leq b <\varepsilon,
  \end{equation*}
  it holds that
  \begin{equation*}
    \Tr[\boldsymbol{1}_{[a,b]}(\mathcal{H}) \varphi] \leq C_1 | b-a| \hbar^{-d} + C_2\hbar^{1-d},
  \end{equation*}   
  for two positive constants $C_1$ and $C_2$ and all $\hbar$ in $(0,\hbar_0]$.  The
  constants $C_1$ and $C_2$ depend only on the numbers $R$, $\hbar_0$ and $c$ in  \eqref{A.non_critical_assumption_prop} and on the bounds on the derivatives of $V$ and $\varphi$.\end{prop}
\begin{remark}
We suppose we have an operator $\mathcal{H}$ acting in $L^2(\R^d)$ with
  $d\geq2$, which obeys Assumption~\ref{A.local_assumptions} with
  $\Omega=B(0,4R)$ for $R>0$. If it is assumed that there exists a $c>0$ for which 
\begin{equation*}
    \abs{V(x)} + \abs{\nabla V(x)}^2 + \hbar \geq c,
 \end{equation*}
for all $x$ in $B(0,2R)$, then by continuity this would imply the existence of a $\tilde{c}>0$ and an $\varepsilon>0$ such that
  \begin{equation*}
    \abs{V(x)-E} + \abs{\nabla V(x)}^2 + \hbar \geq c,
  \end{equation*}
  for all $x$ in $B(0,2R)$ and all $E$ in $[-2\varepsilon,2\varepsilon]$. That is we could generalise the assumptions in the proposition. But we have chosen this form of the proposition due to later applications. 
\end{remark}
\begin{proof}
  We can rewrite the trace of interest as
  \begin{equation}\label{A.Prop_1}
    \begin{aligned}
      \MoveEqLeft \Tr[\boldsymbol{1}_{[a,
        b]}(\mathcal{H})\varphi] =
      \Tr[\boldsymbol{1}_{(-\infty,b]}(\mathcal{H})\varphi]
      - \Tr[\boldsymbol{1}_{(-\infty,a)}(\mathcal{H})\varphi].
    \end{aligned}
  \end{equation}      
  If we consider the trace
  $\Tr[\boldsymbol{1}_{(-\infty,b]}(\mathcal{H}) \varphi]$
  then we can rewrite this in the following way
  \begin{equation*}
    \Tr[\boldsymbol{1}_{(-\infty,b]}(\mathcal{H}) \varphi]=\Tr[\boldsymbol{1}_{(-\infty,0]}(\mathcal{H}-b) \varphi].
  \end{equation*}
  The operator $\mathcal{H}-b$ satisfies
  Assumption~\ref{A.local_assumptions} with $V$ replaced by
  $V-b$ and by assumption we have
  \begin{equation}\label{A.non_critical_assumption_prop_1}
    \abs{V(x)-b} + \abs{\nabla V(x)}^2 + \hbar \geq c,
  \end{equation}
  for all $x$ in $B(0,2R)$. The $b$ should be understood as $b\chi(x)$ where $\chi$ is $C_0^\infty( B(0,4R))$ and $\chi(x)=1$ for $x$ in $B(0,3R)$. Hence we can omit the $\chi$ when we are localised to $B(0,2R)$.  By Theorem~\ref {A.Sobolev_4.1} we
  have the following identity
  \begin{equation}\label{A.Prop_2}
    \begin{aligned}
      \MoveEqLeft
      \Tr[\boldsymbol{1}_{(-\infty,b]}(\mathcal{H})
      \varphi] = \Tr[\boldsymbol{1}_{(-\infty,0]}(\mathcal{H} -
      b) \varphi]
      \\
      &= \frac{1}{(2\pi\hbar)^d} \int_{\R^d} \int_{\R^d}
      \boldsymbol{1}_{\{p^2 + V(x) -b \leq 0\}} (x,p)
      \varphi(x) \;d{x}d{p} + \mathcal{O}( \hbar^{1-d})
      \\
      &= \frac{1}{(2\pi\hbar)^d} \int_{\R^d} \int_{\R^d}
      \boldsymbol{1}_{\{p^2 + V(x) \leq b\}} (x,p)
      \varphi(x) \;d{x}d{p} + \mathcal{O}( \hbar^{1-d}),
    \end{aligned}
  \end{equation}
  where the error term is independent of $b$. Analogously we get that
  \begin{equation}\label{A.Prop_3}
    \Tr[\boldsymbol{1}_{(-\infty,a]}(\mathcal{H}) \varphi] 
    =  \frac{1}{(2\pi\hbar)^d} \int_{\R^d} \int_{\R^d} \boldsymbol{1}_{\{p^2 + V(x) \leq a\}} 
    (x,p) \varphi(x) \;d{x}d{p} +  \mathcal{O}( \hbar^{1-d}).
  \end{equation}
  Since the two error terms are of the same order we can, when
  subtracting the two traces, add the two error terms and obtain a new
  error term of order $\hbar^{1-d}$. Hence we will consider the
  integral
  \begin{equation}\label{A.prop_phase_space_int}
    \int_{\R^d}\int_{\R^d} \boldsymbol{1}_{\{ p^2 +V(x) \leq b\}}(x,p) \varphi(x) 
    - \boldsymbol{1}_{\{ p^2 +V(x) \leq a\}}(x,p) \varphi(x) \;d{x}d{p}. 
  \end{equation}
  By assumption this integral is finite. In order to evaluate these
  integrals we note that by assumption we are in one of the following two
  cases
  \begin{equation}\label{A.prop_first_case}
    \hbar>\frac{c}{2}
  \end{equation}
  or
  \begin{equation}\label{A.non_critical_assumption_prop_2}
    \abs{V(x)-E} + \abs{\nabla V(x)}^2  \geq \frac{c}{2},
  \end{equation}
  for all $x$ in $B(0,2R)$ and all $E$ in
  $[-2\varepsilon,2\varepsilon]$.  In the first case
  \eqref{A.prop_first_case} we can estimate the integrals by a
  constant and replace $\hbar^{-d}$ by $\hbar^{1-d}$ at the cost of
  $\frac{2}{c}$. For the second case
  \eqref{A.non_critical_assumption_prop_2} we have, by the Coarea formula,
  the equality
  \begin{align}\label{A.Prop_4}
    \begin{split}
      \int_{\R^d}\int_{\R^d} \boldsymbol{1}_{\{ p^2 +V(x) \leq
        b\}}&(x,p) \varphi(x) - \boldsymbol{1}_{\{ p^2
        +V(x) \leq a\}}(x,p) \varphi(x) \;d{x}d{p}
      \\
      &= \int_{a}^{b } \int_{\{p^2+V(x)=E \}}
      \varphi(x) \frac{1}{\abs{(\nabla_x V(x),\nabla_p p^2)}} \;d{S}
      d{E},
    \end{split}
  \end{align}
  where $S$ is the surface measure. By support properties of $\varphi$ and \eqref{A.non_critical_assumption_prop_2} we have that
  \begin{equation}\label{A.Prop_6}
  	\sup_{E\in[-\varepsilon,\varepsilon]} \int_{\{p^2+V(x)=E \}}
      \varphi(x) \frac{1}{\abs{(\nabla_x V(x),\nabla_p p^2)}} \;d{S} \leq C.
  \end{equation}	
 Using \eqref{A.Prop_6} we get
  \begin{align}\label{A.Prop_5}
    \begin{split}
      \int_{a}^{b } &\int_{\{p^2+V(x)=E \}}
      \varphi(x)\frac{1}{\abs{(\nabla_x V(x),\nabla_p p^2)}} \;d{S}
      d{E} \leq    \int_{a}^{b} C d{E} \leq C |b-a|,
    \end{split}
  \end{align}
  where $C$ is the constant from \eqref{A.Prop_6}, which is independent of $a$, $b$ and $\hbar$. By
  combining \eqref{A.Prop_1}, \eqref{A.Prop_2}, \eqref{A.Prop_3},
  \eqref{A.Prop_4} and \eqref{A.Prop_5} we get
  \begin{equation*}
    \Tr[\boldsymbol{1}_{[a,b]}(\mathcal{H}) \varphi] \leq  C_1 | b-a| \hbar^{-d} + C_2\hbar^{1-d} .
  \end{equation*} 
  Which is the desired estimate and this ends the proof.
\end{proof}
The previous proposition gives that we can get the right order in
$\hbar$ of the trace if we consider sufficiently small intervals. This will be
a crucial point in the proof of Theorem~\ref{A.Main_Local_Theorem
  _wnc}.

Furthermore we will be needing a corollary to the
Cwikel-Lieb-Rosenbljum (CLR) bound. This corollary is stated in
\cite[Chapter 4]{MR2583992}.
\begin{corollary}\label{A.clr_cor}
  Let $d\geq1$, $\gamma>0$, $\lambda>0$ and $H=-\Delta + V$ be a
  Schr\"odinger operator acting in $L^2(\R^d)$ with
  $(V+\frac{\lambda}{2})_{-}$ in $L^{\frac{d}{2}+\gamma}(\R^d)$ and
  $V_{+}$ in $L^1_{loc}(\R^d)$. Then
  \begin{equation*}
    \Tr(\boldsymbol{1}_{(-\infty,-\lambda]}(H)) \leq \frac{2^\gamma}{\lambda^\gamma } \frac{1}{(4\pi)^{\frac{d}{2}}} \frac{\Gamma(\gamma)}{\Gamma(\tfrac{d}{2}+\gamma)} \int_{\R^d} (V(x)+\tfrac{\lambda}{2})_{-}^{\frac{d}{2}+\gamma} \,dx,
  \end{equation*}
  where $\Gamma$ is the gamma function.
\end{corollary}
We will use this corollary in the following way.
\begin{remark}\label{A.clr_bound_remark}
  Let $H_\hbar=-\hbar^2\Delta+V$ be a Schr\"{o}dinger operator acting
  in $L^2(\R^d)$ and suppose it satisfies
  Assumption~\ref{A.local_assumptions}. We will later need an a priori
  estimate on the number
  $\Tr(\boldsymbol{1}_{(-\infty,\frac{\varepsilon}{4}]}(H_\hbar))$. To
  obtain this we will consider the operator
  $\widetilde{H}_\hbar= -\hbar^2\Delta+V -
  \tfrac{\varepsilon}{2}$. Clearly,
  \begin{equation}\label{A.local_eq14}
    \Tr(\boldsymbol{1}_{(-\infty,-\tfrac{\varepsilon}{4}]}(H_\hbar-\tfrac{\varepsilon}{2})) = \Tr(\boldsymbol{1}_{(-\infty,-\tfrac{\varepsilon}{4\hbar^2}]}(-\Delta + \tfrac{V}{\hbar^2}-\tfrac{\varepsilon}{2\hbar^2})).
  \end{equation}
  If we apply Corollary~\ref{A.clr_cor} to the right hand side of
  \eqref{A.local_eq14} with $\gamma=1$ and
  $\lambda=\tfrac{\varepsilon}{4\hbar^2}$ we get
  \begin{align}\label{A.local_eq15}
    \begin{split}
      \Tr(\boldsymbol{1}_{(-\infty,-\tfrac{\varepsilon}{4\hbar^2}]}(-\Delta
      + \tfrac{V}{\hbar^2}-\tfrac{\varepsilon}{2\hbar^2})) &\leq c_d
      \frac{\hbar^{2}}{\varepsilon} \int_{\R^d}
      (\tfrac{V(x)}{\hbar^2}-\tfrac{3\varepsilon}{8
        \hbar^2})_{-}^{\frac{d}{2}+1} \,dx
      \\
      &= \frac{ c_d}{\varepsilon \hbar^d} \int_{\R^d} (V(x)-
      \tfrac{3\varepsilon}{8})_{-}^{\frac{d}{2}+1} \,dx.
    \end{split}
  \end{align}
  The last integral in \eqref{A.local_eq15} is finite by
  Assumption~\ref{A.local_assumptions} since the support of
  $(V(x)-\tfrac{3\varepsilon}{8})_{-}$ is compact and the function is
  continuous. Combining \eqref{A.local_eq14} with \eqref{A.local_eq15}
  we get the bound
  \begin{equation}\label{A.local_eq16}
    \Tr(\boldsymbol{1}_{(-\infty,\tfrac{\varepsilon}{4}]}(H_\hbar)) \leq \frac{C}{\hbar^{d}}.
  \end{equation}
  where we have absorbed the integral and $\varepsilon$ into the
  constant.
\end{remark}


\subsection{Trace norm estimates of operators}

In this subsection we will list some results on trace norms and
estimates of trace norms for operators. First recall that for an
operator $A$ the trace norm is
\begin{equation*}
  \norm{A}_1=\Tr\left([AA^*]^\frac{1}{2}\right)
\end{equation*}
and the Hilbert-Schmidt norm is
\begin{equation*}
  \norm{A}_2=\sqrt{\Tr\left(AA^*\right)}
\end{equation*}
Moreover we will use the convention that $\norm{A}$ is the operator
norm of $A$.  The following lemma is a modification of
\cite[Lemma~3.9]{MR1343781}. The proofs are completely analogous.
\begin{lemma}\label{A.trace_norm_est_H}
  Let $H_\hbar=-\hbar^2\Delta + V$ be a Schr\"odinger operator acting
  in $L^2(\R^d)$ with $V$ in $ C_0^\infty(\R^d)$. Let $f$ be in
  $C_0^\infty(\R)$ and $\varphi$ in $C_0^\infty(\R^d)$. We let
  $r\in\{0,1\}$, $\hbar_0>0$ and $Q_j=-i\hbar\partial_{x_j}$ for
  $j \in\{1,\dots,d\}$. Then
  \begin{equation*}
    \norm{\varphi Q_j^r f(H_\hbar)}_1 \leq C \hbar^{-d},
  \end{equation*}
  for all $\hbar$ in $(0,\hbar_0]$. If $\psi$ is a bounded function from $C^\infty(\R^d)$ and $c>0$ such
  that
  \begin{equation}\label{A.lemme_dep_c}
    \mathrm{dist}[\supp(\varphi),\supp(\psi)] \geq c.
  \end{equation}
  Then for any $N$ in $\N$
  \begin{equation*}
    \norm{\varphi Q_j^r f(H_\hbar)\psi}_1 \leq C_N \hbar^{N},
  \end{equation*}	
  for all $\hbar$ in $(0,\hbar_0]$. Both constants $C$ and $C_N$ depend on the dimension, the
  functions $\varphi$ and $\psi$, the numbers $\hbar_0$, $\norm{\partial^\alpha V}_\infty$ for $\alpha$ in $\N^d$,
  $\norm{\partial^j f}_\infty$ for $j$ in $\N$, $c$ in \eqref{A.lemme_dep_c} and $\sup(\supp(f))$.
\end{lemma}
The following theorem is an extension of Theorem~3.12 from the paper
\cite{MR1343781} as an extra operator has been added. It is less
general in the sense that we only consider compactly supported,
smooth functions applied to the operator, whereas in the paper more
general functions are considered. Again we omit the easy modifications of the proof in \cite{MR1343781}.
\begin{thm}\label{A.change_operator}
  Let $\mathcal{H}$ satisfy Assumption~\ref{A.local_assumptions} with
  $\Omega=B(0,4R)$ for an $R>0$. Let $f$ be in $C_0^\infty(\R)$ and
  let $r\in\{0,1\}$, $\hbar_0>0$ and $Q_j=-i\hbar\partial_{x_j}$ for
  $j \in\{1,\dots,d\}$. If $\varphi$ is in $C_0^\infty(B(0,3R))$ then
  for any $N\geq0$
  \begin{align*}
    \norm{\varphi Q_j^r[f(\mathcal{H})-f(H_\hbar)]}_1 &\leq C_N
    \hbar^N \intertext{and} \norm{\varphi Q_j^r f(\mathcal{H})}_1
    &\leq C\hbar^{-d}
  \end{align*}
 for all $\hbar$ in $(0,\hbar_0]$, where the constants $C_N$ and $C$ only depend on the dimension and the numbers $\hbar_0$,
  $\norm{\partial^j f}_\infty$ for $j$ in $\N$ and $\norm{\partial^\alpha V}_\infty$, $\norm{\partial^\alpha\varphi}_\infty$ for $\alpha$ in $\N^d$.
\end{thm}


\section{Local case}

In this section we will present the first step in the proof of
Theorem~\ref{A.Main_Theorem} where we prove a local version of the
theorem under a non-critical condition. It should be noted that we are
not trying to get optimal constants in the following.

\subsection{Auxiliary bounds}
Before we proceed we will consider a simple case where the function
applied to the operator is a smooth function with compact
support. Moreover we will prove a bound on a Hilbert-Schmidt norm
which will prove to be useful.

The first auxiliary result is a simple case of
Theorem~\ref{A.Main_Local_Theorem _wnc}, where we consider the same
commutators as in the theorem but we apply a smooth, compactly
supported function to our operator instead of the characteristic
function.

\begin{lemma}\label{A.simpel_case}
  Suppose the operator $\mathcal{H}$ obeys
  Assumption~\ref{A.local_assumptions} with $\Omega=B(0,4R)$ for $R>0$
  and let $f$ be in $C_0^\infty(\R)$ and $\hbar_0>0$. For $\varphi$ in
  $C_0^\infty(B(0,3R))$ and $Q_j=-i\hbar\partial_{x_j}$ for
  $j \in\{1,\dots,d\}$ it holds that
  \begin{equation*}
    \norm{[f(\mathcal{H}), \varphi]}_{1} \leq C \hbar^{1-d} \qquad\text{and}\qquad  \norm{[f(\mathcal{H}), \varphi Q_j]}_{1} \leq C\hbar^{1-d},
  \end{equation*}
  for all $\hbar$ in $(0,\hbar_0]$ and a positive constant $C$,
  where $C$ only depend on the dimension, the function $\varphi$, the
  numbers $\hbar_0$, $\norm{\partial^\alpha V}_\infty$ for $\alpha$ in $\N^d$, $\norm{\partial^j f}_\infty$ for $j$
  in $\N$ and $\sup(\supp(f))$.
\end{lemma}
\begin{proof}
  We start by proving the first commutator bound. By
  Theorem~\ref{A.change_operator} we note that for any $N\geq0$
  \begin{equation}\label{A.Estimate_1.1}
    \norm{[f(\mathcal{H}), \varphi]}_{1} \leq \norm{[f(H_\hbar), \varphi]}_{1} + C_N\hbar^N,
  \end{equation}
  hence we need only prove the bound for the trace norm of
  $[f(H_\hbar), \varphi]$. Let $g\in C_0^\infty(\R)$ such that
  $g(t)f(t)=f(t)$ and $0\leq g(t) \leq 1$ for all $t$ in $\R$. Then we
  have that
  \begin{align*}
    [f(H_\hbar), \varphi] &= f(H_\hbar) \varphi - \varphi f(H_\hbar)
    \\
    &= g(H_\hbar)f(H_\hbar) \varphi - \varphi g(H_\hbar) f(H_\hbar) +
    g(H_\hbar)\varphi f(H_\hbar) - g(H_\hbar)\varphi f(H_\hbar)
    \\
    &=g(H_\hbar) [f(H_\hbar), \varphi] + [g(H_\hbar), \varphi]
    f(H_\hbar).
  \end{align*}
  These equalities implie that
  \begin{equation}\label{A.Estimate_1.2}
    \norm{[f(H_\hbar), \varphi]}_{1} \leq \norm{g(H_\hbar) [f(H_\hbar), \varphi]}_1 +  \norm{[g(H_\hbar), \varphi] f(H_\hbar)}_1.
  \end{equation}
  We start by considering the first trace norm
  $\norm{g(H_\hbar) [f(H_\hbar), \varphi]}_1$ and the second can be
  treated by an analogous argument. Let $\widetilde{\varphi}$ be in
  $C_0^\infty(\R^d)$ such that $\widetilde{\varphi}\varphi=\varphi$ and
  $0\leq \widetilde{\varphi}\leq1$. Then we have that
  \begin{align}\label{A.Estimate_1.3}
    \begin{split}
      \norm{g(H_\hbar) [f(H_\hbar), \varphi]}_1 &\leq
      \norm{g(H_\hbar)\widetilde{\varphi} [f(H_\hbar), \varphi]}_1 +
      \norm{g(H_\hbar)(1-\widetilde{\varphi}) f(H_\hbar) \varphi}_1
      \\
      &\leq \norm{g(H_\hbar)\widetilde{\varphi} }_1 \norm{[f(H_\hbar),
        \varphi]}+ \norm{(1-\widetilde{\varphi}) f(H_\hbar) \varphi}_1
      \\
      &\leq C\hbar^{-d} \norm{[f(H_\hbar), \varphi]} + C_N \hbar^N,
    \end{split}
  \end{align}
  for all $N\geq 0$, where we have used Lemma~\ref{A.trace_norm_est_H}
  in the last inequality. That
  \begin{equation}
  	\norm{[f(H_\hbar), \varphi]}\leq C\hbar,
\end{equation}
 is a consequence of the
  functional calculus for $\hbar$-$\Psi$DOs presented in
  \cite{MR897108}. It also follows fairly easily from an argument
  using the Helffer-Sj\"ostrand formula \cite{MR1349825} and the
  resolvent identities. The estimate on the second term in
  \eqref{A.Estimate_1.2} is similar and will be left to the
  reader. This estimate concludes the proof.
\end{proof}

The next lemma is very similar to the above lemma.

\begin{lemma}\label{A.Est_help}
  Suppose the operator $\mathcal{H}$ obeys
  Assumption~\ref{A.local_assumptions} with $\Omega=B(0,4R)$ for $R>0$
  and let $f$ be in $C_0^\infty(\R)$ and $\hbar_0>0$. For $\varphi$ in
  $C_0^\infty(B(0,3R))$ it holds that
  \begin{equation*}
    \norm{[\mathcal{H},\varphi]f(\mathcal{H})}_2 \leq C\hbar^{1-\frac{d}{2}},
  \end{equation*}
  for all $\hbar$ in $(0,\hbar_0]$ for a positive constant $C$, where $C$ only depends on the dimension,
  the function $\varphi$, the numbers $\hbar_0$, $\norm{\partial^\alpha V}_\infty$ for $\alpha$ in $\N^d$,
  $\norm{\partial^j f}_\infty$ for $j$ in $\N$ and $\sup(\supp(f))$.
\end{lemma}

\begin{proof}
  Let $\varphi_1$ be in $C_0^\infty(B(0,3R))$ such that
  $\varphi_1\varphi=\varphi$ and $0\leq \varphi_1\leq1$. Then by
  Assumption~\ref{A.local_assumptions} the commutator
  $[\mathcal{H},\varphi]$ is local in the sense that
  \begin{equation*}
    [\mathcal{H},\varphi]=[H_\hbar,\varphi]\varphi_1,
  \end{equation*}
  where $H_\hbar$ is the operator from
  Assumption~\ref{A.local_assumptions} i.e.
  $H_\hbar=-\hbar^2\Delta + V$, where $V$ is in
  $C_0^\infty(\R^d)$. Therefore there exists a $\lambda_0\geq0$ such
  that $-\lambda_0$ is in the resolvent set of $H_\hbar$ and the
  operator $H_\hbar+\lambda_0$ is positive
  (e.g. $\lambda_0=1+\norm{V}_\infty$) We then have that
  \begin{align}\label{A.Estimate_2.1}
    \begin{split}
      \norm{[\mathcal{H},\varphi]f(\mathcal{H})}_2 =&
      \norm{[H_\hbar,\varphi]\varphi_1R_{H_\hbar}(-\lambda_0)
        (H_\hbar+\lambda_0)f(\mathcal{H})}_2
      \\
      \leq& \norm{[H_\hbar,\varphi] R_{H_\hbar}(-\lambda_0)\varphi_1
        (H_\hbar+\lambda_0) f(\mathcal{H}) }_2
      \\
      \hbox{} &+ \norm{[H_\hbar,\varphi] [\varphi_1 , R_{H_\hbar}
        (-\lambda_0)] (H_\hbar+\lambda_0)f(\mathcal{H})}_2,
    \end{split}
  \end{align}
  where $R_{H_\hbar}(z)=(H_\hbar-z)^{-1}$. If we now consider each of
  the terms separately we can for the first term note that by
  Assumption~\ref{A.local_assumptions} and
  Theorem~\ref{A.change_operator} we have
  \begin{align}\label{A.Estimate_2.2}
    \begin{split}
      \norm{[H_\hbar,\varphi] R_{H_\hbar}(-\lambda_0)\varphi_1
        (H_\hbar+\lambda_0) f(\mathcal{H}) }_2 &\leq
      \norm{[H_\hbar,\varphi] R_{H_\hbar}(-\lambda_0)} \norm{\varphi_1
        (\mathcal{H}+\lambda_0) f(\mathcal{H}) }_2
      \\
      &\leq c \hbar^{-\frac{d}{2}} \norm{[H_\hbar,\varphi]
        R_{H_\hbar}(-\lambda_0)}
      \\
      &\leq C \hbar^{1-\frac{d}{2}},
    \end{split}
  \end{align}
  where we have used the bound
  \begin{equation}\label{A.bound_com_proof_1}
    \norm{[H_\hbar,\varphi] R_{H_\hbar}(-\lambda_0)}\leq \hbar \sum_{j=1}^d \norm{(2\varphi_{x_j} Q_j - i\hbar\varphi_{x_j x_j})R_{H_\hbar}(-\lambda_0) } \leq  c \hbar,
  \end{equation}
  where we have calculated the commutator explicitly. The bound in \eqref{A.bound_com_proof_1} is valid since $\mathcal{D}(H_\hbar)\subset \mathcal{D}(Q_j)$ for all
  $j\in\{1,\dots,d\}$. Moreover in \eqref{A.Estimate_2.2} we have used the
  following estimate
  \begin{align*}
    \norm{\varphi_1 (\mathcal{H}+\lambda_0) f(\mathcal{H}) }_2^2 &=
    \Tr[\varphi_1 (\mathcal{H}+\lambda_0) f(\mathcal{H})^2
    (\mathcal{H}+\lambda_0)\varphi_1 ]
    \\
    &\leq \norm{\varphi_1 (\mathcal{H}+\lambda_0) f(\mathcal{H})^2
      (\mathcal{H}+\lambda_0)\varphi_1}_1 \leq C\hbar^{-d},
  \end{align*}
  by Theorem~\ref{A.change_operator}. For the other term on the right
  hand side of \eqref{A.Estimate_2.1} we note that
  \begin{equation}\label{A.Estimate_2.3}
    \norm{[H_\hbar,\varphi] [\varphi_1 , R_{H_\hbar} (-\lambda_0)] (H_\hbar+\lambda_0)f(\mathcal{H})}_2 = \norm{[H_\hbar,\varphi] R_{H_\hbar}(-\lambda_0) [H_\hbar,\varphi_1 ] f(\mathcal{H})}_2
  \end{equation}
  Let $\varphi_2$ be in $C_0^\infty(B(0,3R))$ such that
  $\varphi_2\varphi_1=\varphi_1$ and $0\leq \varphi_2\leq1$ and note
  that by Theorem~\ref{A.change_operator}
  \begin{align}\label{A.Estimate_2.3}
  	\begin{split}
    \lVert[H_\hbar,\varphi] R_{H_\hbar}(-\lambda_0)
    &[H_\hbar,\varphi_1 ] f(\mathcal{H})\rVert_2
    \\
    &= \norm{[H_\hbar,\varphi] R_{H_\hbar}(-\lambda_0)
      [H_\hbar,\varphi_1 ]\varphi_2 f(\mathcal{H})}_2
    \\
    &\leq \norm{[H_\hbar,\varphi] R_{H_\hbar}(-\lambda_0)^\frac{1}{2}}
    \norm{R_{H_\hbar}(-\lambda_0)^\frac{1}{2} [H_\hbar,\varphi_1
      ]}\norm{\varphi_2 f(\mathcal{H})}_2
    \\
    &\leq C \hbar^{2-\frac{d}{2}},
    \end{split}
  \end{align}
  where we have used that the commutators $[H_\hbar,\varphi]$ and
  $[H_\hbar,\varphi_1]$ can be calculated explicitly and that their domains
  contains the form domain of $H_\hbar$. Combining estimates \eqref{A.Estimate_2.1}, \eqref{A.Estimate_2.2} and \eqref{A.Estimate_2.3}  we
  get the desired bound.
\end{proof}


\subsection{Local case with a non-critical condition}

We will now state and prove the local version of the main theorem
(Theorem~\ref{A.Main_Theorem}) with a non-critical condition. It
should be noted that we are only dealing with open balls as the domain
in Assumption~\ref{A.local_assumptions} since when we extend the
result we will use them to cover a general open set.

\begin{thm}\label{A.Main_Local_Theorem _wnc}
  Suppose the operator $\mathcal{H}$ acting in $L^2(\R^d)$ with
  $d\geq2$ obeys Assumption~\ref{A.local_assumptions} with
  $\Omega=B(0,4R)$ for $R>0$ and
  \begin{equation}\label{A.noncrit.Main_Local_Theorem_wnc}
    \abs{V(x)}+ \abs{\nabla V(x)}^2 + \hbar  \geq c,
  \end{equation}
  for all $x$ in $B(0,2R)$, where $c>0$. Furthermore, let $\hbar_0$ be a strictly positive number. For $\varphi$ in
  $C_0^\infty(B(0,R/2))$ it holds that
  \begin{equation}\label{A.local_bounds}
    \norm{[\boldsymbol{1}_{(-\infty,0]}(\mathcal{H}), \varphi]}_{1} \leq  C \hbar^{1-d}
    \quad\text{and}\quad
    \norm{[\boldsymbol{1}_{(-\infty,0]}(\mathcal{H}), \varphi Q_j]}_{1} \leq  C \hbar^{1-d},
  \end{equation}
  for all $\hbar$ in $(0,\hbar_0]$ and $j\in\{1,\dots,d\}$, where
  $Q_j= -i\hbar\partial_{x_j}$. The constant $C$ only depends on
  the dimension, the numbers $\norm{\partial_x^\alpha V}_\infty$ and $\norm{\partial_x^\alpha \varphi}_\infty$
  for all $\alpha$ in $\N^d$, and the numbers $R$ and $c$ in
  \eqref{A.noncrit.Main_Local_Theorem_wnc}.
\end{thm}
\begin{proof}
  We start by proving the first bound in \eqref{A.local_bounds}. We
  notice that
  \begin{equation}\label{A.local_eq_1}
    [\boldsymbol{1}_{(-\infty,0]}(\mathcal{H}),\varphi] = \boldsymbol{1}_{(-\infty,0]}(\mathcal{H}) \varphi \boldsymbol{1}_{(0,\infty)}(\mathcal{H}) - \boldsymbol{1}_{(0,\infty)}(\mathcal{H}) \varphi \boldsymbol{1}_{(-\infty,0]}(\mathcal{H}).
  \end{equation}
  We will consider each of the terms in  \eqref{A.local_eq_1} separately and
  they can be handled with analogous arguments. So we only consider
  the term
  $\boldsymbol{1}_{(-\infty,0]}(\mathcal{H}) \varphi
  \boldsymbol{1}_{(0,\infty)}(\mathcal{H})$. By \eqref{A.noncrit.Main_Local_Theorem_wnc} and continuity, there
  exists an $\varepsilon>0$ such that for all $E$ in
  $[-2\varepsilon,2\varepsilon]$ we have
  \begin{equation*}
    \abs{E-V(x)}+ \abs{\nabla V(x)}^2 + \hbar  \geq \frac{c}{2}, 
  \end{equation*}
  for all $x$ in $B(0,2R)$. Without loss of generality we can assume $\varepsilon\leq1$. Let $g_1$ and $g_0$ be two functions such
  that
  \begin{itemize}
  \item
    $g_1(\mathcal{H}) + g_0(\mathcal{H}) =
    \boldsymbol{1}_{(-\infty,0]}(\mathcal{H})$.
	
  \item $\supp(g_0) \subset [-\varepsilon,0]$ and $g_0(t)=1$ for
    $t \in [-\varepsilon/2,0]$.
	
  \item $g_1 \in C_0^\infty(\R)$.
  \end{itemize}
  That $g_1$ can assumed to be compactly supported is due to the fact
  that the spectrum of $\mathcal{H}$ is bounded from below. With these
  functions we get that
  \begin{align}\label{A.local_eq_2}
    \begin{split}
      \boldsymbol{1}_{(-\infty,0]}(\mathcal{H}) \varphi
      \boldsymbol{1}_{(0,\infty)}(\mathcal{H}) &= g_1(\mathcal{H})
      \varphi \boldsymbol{1}_{(0,\infty)}(\mathcal{H}) +
      g_0(\mathcal{H}) \varphi
      \boldsymbol{1}_{(0,\infty)}(\mathcal{H})
      \\
      &= [g_1(\mathcal{H}), \varphi]
      \boldsymbol{1}_{(0,\infty)}(\mathcal{H}) + g_0(\mathcal{H})
      \varphi \boldsymbol{1}_{(0,\infty)}(\mathcal{H}).
    \end{split}
  \end{align}
  For the first term we note that by Lemma~\ref{A.simpel_case} we have
  the estimate:
  \begin{equation}\label{A.local_eq_3}
    \norm{ [g_1(\mathcal{H}), \varphi]  \boldsymbol{1}_{(0,\infty)}(\mathcal{H})  }_{1} \leq \norm{ [g_1(\mathcal{H}), \varphi] }_{1} 
    \leq C \hbar^{1-d}.
  \end{equation}
  In order to estimate the term
  $g_0(\mathcal{H}) \varphi \boldsymbol{1}_{(0,\infty)}(\mathcal{H})$
  we let $f$ be in $C_0^\infty(\R)$ such that $f(t)=1$ on
  $[-\varepsilon,0]$ and
  $\supp(f) \subset[-2\varepsilon,\varepsilon]$. Then we have
  \begin{align*}
    g_0(\mathcal{H}) \varphi \boldsymbol{1}_{(0,\infty)}(\mathcal{H})
    &= g_0(\mathcal{H})f(\mathcal{H}) \varphi
    \boldsymbol{1}_{(\varepsilon,\infty)}(\mathcal{H}) +
    g_0(\mathcal{H}) \varphi
    \boldsymbol{1}_{(0,\varepsilon]}(\mathcal{H})
    \\
    & = g_0(\mathcal{H}) [f(\mathcal{H}), \varphi ]
    \boldsymbol{1}_{(\varepsilon,\infty)}(\mathcal{H}) +
    g_0(\mathcal{H}) \varphi
    \boldsymbol{1}_{(0,\varepsilon]}(\mathcal{H}).
  \end{align*}
  Again from Lemma~\ref{A.simpel_case} we have the estimate:
  \begin{equation}\label{A.local_eq_13}
    \norm{  g_0(\mathcal{H}) [f(\mathcal{H}),  \varphi ] \boldsymbol{1}_{(\varepsilon,\infty)}(\mathcal{H})   }_{1}
    \leq \norm{ [f(\mathcal{H}),  \varphi ] }_{1} 
    \leq C\hbar^{1-d}.
  \end{equation}
  What remains is to get an estimate of the trace norm of the term
  $g_0(\mathcal{H}) \varphi
  \boldsymbol{1}_{(0,\varepsilon]}(\mathcal{H})$. In order to estimate
  this term we define the following $\hbar$ dependent dyadic
  decomposition:
  \begin{equation*}
    \chi_{n,\hbar} (t) = 
    \begin{cases} 
      \boldsymbol{1}_{(\hbar,0]}(t) & n=0
      \\
      \boldsymbol{1}_{(-4^n\hbar,-4^{n-1}\hbar]}(t) & n\in\N .
    \end{cases}
  \end{equation*}
  moreover we let $\widetilde{\chi}_{n,\hbar}(t) = \chi_{n,\hbar} (-t)$.
  Then there exist $N(\hbar)$ in $\N$ such that
  \begin{equation*}
    g_0(\mathcal{H}) = \sum_{n=0}^{N(\hbar)} g_0(\mathcal{H})\chi_{n,\hbar}(\mathcal{H}) 
    \quad\text{and}\quad 
    \boldsymbol{1}_{(0,\varepsilon]}(\mathcal{H}) = \sum_{m=0}^{N(\hbar)} \boldsymbol{1}_{(0,\varepsilon]}(\mathcal{H}).
    \widetilde{\chi}_{m,\hbar}(\mathcal{H}).
  \end{equation*}
  With these equalities we get the following inequality:
  \begin{align}\label{A.Full_sum}
    \begin{split}
      \lVert g_1 &(\mathcal{H}) \varphi
      \boldsymbol{1}_{(0,\varepsilon]}(\mathcal{H})\rVert_{1} \leq
      \sum_{n=0}^{N(\hbar)} \sum_{m=0}^{N(\hbar)}
      \norm{\chi_{n,\hbar}(\mathcal{H}) \varphi
        \widetilde{\chi}_{m,\hbar}(\mathcal{H})}_{1}
      \\
      = & \sum_{n=1}^{N(\hbar)}\sum_{m\geq n}^{N(\hbar)}
      \norm{\chi_{n,\hbar}(\mathcal{H}) \varphi
        \widetilde{\chi}_{m,\hbar}(\mathcal{H})}_{1} +
      \sum_{m=1}^{N(\hbar)}\sum_{n> m}^{N(\hbar)}
      \norm{\chi_{n,\hbar}(\mathcal{H}) \varphi
        \widetilde{\chi}_{m,\hbar}(\mathcal{H})}_{1}
      \\
      & \hbox{} + \sum_{n=1}^{N(\hbar)}
      \norm{\chi_{n,\hbar}(\mathcal{H}) \varphi
        \widetilde{\chi}_{0,\hbar}(\mathcal{H})}_{1}
      +\sum_{m=1}^{N(\hbar)}\norm{\chi_{0,\hbar}(\mathcal{H}) \varphi
        \widetilde{\chi}_{m,\hbar}(\mathcal{H})}_{1} +
      \norm{\chi_{0,\hbar}(\mathcal{H}) \varphi
        \widetilde{\chi}_{0,\hbar}(\mathcal{H})}_{1}.
    \end{split}
  \end{align}
  We will start by considering a term from the first double sum.
  Hence we assume that $m\geq n>0$. The support of
  $\chi_{n,\hbar}(\mathcal{H})$ is
  $[-4^n \hbar ,-4^{n-1} \hbar] = [-2^{2n}\hbar,-2^{2(n-1)}\hbar]$,
  which contains the point $-2^{2n-1}\hbar$, and similarly the support
  of $\widetilde{\chi}_{m,\hbar}(\mathcal{H})$ is
  $[4^{m-1}\hbar,4^{m}\hbar] = [2^{2(m-1)}\hbar,2^{2m}\hbar]$. We note
  that we can make the following estimate, using the spectral theorem.
  \begin{align*}
    \MoveEqLeft[3] \lVert\chi_{n,\hbar}(\mathcal{H}) \varphi
    \widetilde{\chi}_{m,\hbar}(\mathcal{H})\rVert_{1}
    \\
    ={}& \norm{\chi_{n,\hbar}(\mathcal{H}) \varphi (\mathcal{H}
      +2^{2n-1}\hbar) \widetilde{\chi}_{m,\hbar}(\mathcal{H}) (\mathcal{H}
      + 2^{2n-1}\hbar)^{-1}}_{1}
    \\
    \leq{} & \norm{\chi_{n,\hbar}(\mathcal{H})\varphi (\mathcal{H} +
      2^{2n-1}\hbar)
      \widetilde{\chi}_{m,\hbar}(\mathcal{H})}_{1}(2^{2(m-1)}\hbar+
    2^{2n-1}\hbar)^{-1}
    \\
    \leq{} &(2^{2(m-1)}\hbar+ 2^{2n-1}\hbar)^{-1} \big\{
    \norm{\chi_{n,\hbar}(\mathcal{H}) (\mathcal{H} +
      2^{2n-1}\hbar)\varphi \widetilde{\chi}_{m,\hbar}(\mathcal{H})}_{1}
    \\
    &+ \norm{\chi_{n,\hbar} (\mathcal{H}) [\varphi,\mathcal{H}]
      \widetilde{\chi}_{m,\hbar}(\mathcal{H})}_{1} \big\}
    \\
    \leq{} & \frac{2^{2n-1}}{2^{2(m-1)}+ 2^{2n-1}}
    \norm{\chi_{n,\hbar}(\mathcal{H}) \varphi
      \widetilde{\chi}_{m,\hbar}(\mathcal{H})}_{1}
    \\
    &+ (2^{2(m-1)}\hbar+ 2^{2n-1}\hbar)^{-1} \norm{\chi_{n,\hbar}
      (\mathcal{H}) [\varphi,\mathcal{H}]
      \widetilde{\chi}_{m,\hbar}(\mathcal{H})}_{1},
  \end{align*}
  With
  \begin{equation*}
    a \coloneqq  \frac{2^{2n-1}}{2^{2(m-1)}+ 2^{2n-1}} .
  \end{equation*}
  The above calculation implies that
  \begin{equation*}
    (1-a)  \norm{\chi_{n,\hbar}(\mathcal{H}) \varphi \widetilde{\chi}_{m,\hbar}(\mathcal{H})}_{1} \leq (2^{2(m-1)}\hbar+ 2^{2n-1}\hbar)^{-1}  \norm{\chi_{n,\hbar} (\mathcal{H}) 
      [\varphi,	\mathcal{H}] \widetilde{\chi}_{m,\hbar}(\mathcal{H})}_{1}.
  \end{equation*}
  This implies the following estimate
  \begin{align}\label{A.local_eq_4}
    \begin{split}
      \norm{\chi_{n,\hbar}(\mathcal{H}) \varphi
        \widetilde{\chi}_{m,\hbar}(\mathcal{H})}_{1} \leq& (1-a)^{-1}
      (2^{2(m-1)}\hbar+ 2^{2n-1}\hbar)^{-1} \norm{\chi_{n,\hbar}
        (\mathcal{H}) [\varphi,\mathcal{H}]
        \widetilde{\chi}_{m,\hbar}(\mathcal{H})}_{1}
      \\
      \leq& \frac{1}{2^{2(m-1)}\hbar}\norm{\chi_{n,\hbar}
        (\mathcal{H}) [\varphi,\mathcal{H}]
        \widetilde{\chi}_{m,\hbar}(\mathcal{H})}_{1}
    \end{split}
  \end{align}
  Due to the double sum in \eqref{A.Full_sum} we need to repeat the
  argument. By an analogous argument the following estimate holds
  \begin{equation}\label{A.local_eq_5}
    \norm{\chi_{n,\hbar} (\mathcal{H}) [\varphi,\mathcal{H}] \widetilde{\chi}_{m,\hbar}(\mathcal{H})}_{1}
    \leq \frac{1}{2^{2(m-1)}\hbar} \norm{\chi_{n,\hbar} (\mathcal{H})  
      [[\varphi,\mathcal{H}],\mathcal{H}] \widetilde{\chi}_{m,\hbar}(\mathcal{H})}_{1}.
  \end{equation}
  By combining \eqref{A.local_eq_4} and \eqref{A.local_eq_5} we get
  that
  \begin{equation}\label{A.Estimate_1_III}
    \norm{\chi_{n,\hbar} (\mathcal{H}) \varphi \widetilde{\chi}_{m,\hbar}(\mathcal{H})}_{1}
    \leq  \frac{1}{4^{2(m-1)}\hbar^2} \norm{\chi_{n,\hbar} (\mathcal{H})  
      [[\varphi,\mathcal{H}],\mathcal{H}] \widetilde{\chi}_{m,\hbar}(\mathcal{H})}_{1}.
  \end{equation}
  We will now prove that
  \begin{equation}\label{A.local_eq_12}
    \norm{\chi_{n,\hbar} (\mathcal{H}) \varphi \widetilde{\chi}_{m,\hbar}(\mathcal{H})}_{1}
    \leq C  \frac{4^{\frac{m+n}{2}}  \hbar^{3-d}}{4^{2(m-1)}\hbar^2} = \frac{16 C }{4^{\frac{3}{4}m - \frac{1}{2}n}}   \hbar^{1-d},
  \end{equation}
  for $m\geq n \geq1$ is true. By Assumption~\ref{A.local_assumptions}
  we have that
  \begin{equation}\label{A.local_eq_21}
    [[\varphi,\mathcal{H}],\mathcal{H}] = [[\varphi,H_\hbar],H_\hbar],
  \end{equation}
  since we have assumed that the operator $\mathcal{H}$ acts on
  $C_0^\infty(B(0,4R))$ as the operator $H_\hbar$. By a calculation we
  note that
  \begin{equation}\label{A.local_eq_6}
    [[\varphi,H_\hbar],H_\hbar] 
    =   \hbar^2 \sum_{j=1}^d  \sum_{l=1}^d \big[  -2( Q_l \varphi_{x_j x_l}  Q_j  + Q_j \varphi_{x_j x_l}  Q_l) + 2 \varphi_{x_j} V_{x_j} + \hbar^2  \varphi_{x_j x_j x_l x_l } \big],
  \end{equation}
  where $Q_j=-i\hbar\partial_{x_j}$ and
  $\varphi_{x_j}(x)=(\partial_{x_j}\varphi)(x)$. With this form of the
  double commutator we have
  \begin{equation}\label{A.local_eq_7}
    \norm{ R_{H_\hbar}(i) [[\varphi,\mathcal{H}],\mathcal{H}] R_{H_\hbar}(i) } \leq c \hbar^2,
  \end{equation}
  where $R_{H_\hbar}(i)=(H_\hbar-i)^{-1}$ is the resolvent at the
  point $i$, since $\mathcal{D}(H_\hbar)\subset \mathcal{D}(Q_j)$ for
  all $j\in\{1,\dots,d\}$. In order to estimate the right hand side in
  \eqref{A.local_eq_5} let $\psi$ be in $C_0^\infty(\R^d)$ such that
  $\psi(x)=1$ for all $x$ in $\supp(\varphi)$ and
  $\supp(\psi)\subset B(0,R/2)$. As the double commutator is local,
  which follows from \eqref{A.local_eq_21} and \eqref{A.local_eq_6}, we have
  \begin{equation}\label{A.local_eq_8}
    \norm{\chi_{n,\hbar} (\mathcal{H})[[\varphi,\mathcal{H}],\mathcal{H}] \widetilde{\chi}_{m,\hbar}(\mathcal{H})}_{1}
    =
    \norm{\chi_{n,\hbar} (\mathcal{H}) \psi [[\varphi,\mathcal{H}],\mathcal{H}] \psi \widetilde{\chi}_{m,\hbar}(\mathcal{H})}_{1}.
  \end{equation}
  By inserting two resolvents, applying a Cauchy-Schwarz inequality
  and the estimate \eqref{A.local_eq_7}, we have
  \begin{equation}\label{A.local_eq_9}
    \begin{aligned}
      \MoveEqLeft \lVert\chi_{n,\hbar} (\mathcal{H}) \psi
      [[\varphi,\mathcal{H}],\mathcal{H}] \psi
      \widetilde{\chi}_{m,\hbar}(\mathcal{H})\rVert_{1}
      \\
      &= \norm{\chi_{n,\hbar} (\mathcal{H}) \psi(H_\hbar-i)
        R_{H_\hbar}(i) [[\varphi,\mathcal{H}],\mathcal{H}]
        R_{H_\hbar}(i) (H_\hbar-i) \psi
        \widetilde{\chi}_{m,\hbar}(\mathcal{H})}_{1}
      \\
      &\leq c \hbar^2 \norm{\chi_{n,\hbar} (\mathcal{H})
        \psi(H_\hbar-i)}_{2} \norm{(H_\hbar-i) \psi
        \widetilde{\chi}_{m,\hbar}(\mathcal{H})}_{2}.
    \end{aligned}
  \end{equation}
  If we consider the first of the two Hilbert-Schmidt norms we have
  \begin{equation}\label{A.local_eq_10}
    \norm{\chi_{n,\hbar} (\mathcal{H}) \psi(H_\hbar-i)}_{2} \leq  \norm{\chi_{n,\hbar} (\mathcal{H}) (H_\hbar-i) \psi}_{2} +  \norm{\chi_{n,\hbar} (\mathcal{H}) [\psi,H_\hbar]}_{2}.
  \end{equation}
  By Assumption~\ref{A.local_assumptions} and
  Proposition~\ref{A.Weyl_used_local} we have
  \begin{equation}\label{A.local_eq_11}
    \norm{\chi_{n,\hbar} (\mathcal{H}) (H_\hbar-i) \psi}_{2} \leq 2  \norm{\chi_{n,\hbar} (\mathcal{H}) \psi}_{2} = 2 \sqrt{\Tr[\psi \chi_{n,\hbar} (\mathcal{H})^2 \psi ] }\leq 2 \sqrt{C 4^n \hbar^{1-d}}.
  \end{equation}
  For the second term in \eqref{A.local_eq_10} let $f$ be in
  $C_0^\infty(\R)$ such that $f(t)=1$ for
  $t\in[-\frac32 \varepsilon,\frac32 \varepsilon]$ and $f(t)=0$ for
  $\abs{t}\geq 2\varepsilon$. Then we have the bound
  \begin{equation*}
    \norm{\chi_{n,\hbar} (\mathcal{H}) [\psi,H_\hbar] }_{2}= \norm{\chi_{n,\hbar} (\mathcal{H}) f(\mathcal{H}) [\psi,H_\hbar] }_{2} \leq \norm{ f(\mathcal{H}) [\psi,H_\hbar]}_{2} \leq c\hbar^{1-\frac{d}{2}}.
  \end{equation*}
  by Lemma~\ref{A.Est_help}. Combining this estimate with
  \eqref{A.local_eq_10} and \eqref{A.local_eq_11} we get
  \begin{equation}\label{A.local_eq_17}
    \norm{\chi_{n,\hbar} (\mathcal{H}) \psi(H_\hbar-i)}_{2} \leq   \sqrt{\widetilde{C} 4^n \hbar^{1-d}}.
  \end{equation}
  By analogous estimates we also get
  \begin{equation}\label{A.local_eq_18}
    \norm{(H_\hbar-i) \psi \widetilde{\chi}_{m,\hbar}(\mathcal{H})}_{2}  \leq   \sqrt{\widetilde{C} 4^m \hbar^{1-d}}.
  \end{equation}
  Now by combing\eqref{A.local_eq_17} and \eqref{A.local_eq_18}  with \eqref{A.local_eq_8} and
  \eqref{A.local_eq_9} we get
  \begin{equation}\label{A.Estimate_3_III}
    \norm{\chi_{n,\hbar} (\mathcal{H})[[\varphi,\mathcal{H}],\mathcal{H}] \widetilde{\chi}_{m,\hbar}(\mathcal{H})}_{1}
    \leq
    C 4^{\frac{n+m}{2}} \hbar^{1-d}.
  \end{equation}
  By \eqref{A.Estimate_1_III} and \eqref{A.Estimate_3_III} we have
  the estimate \eqref{A.local_eq_12}. Using \eqref{A.local_eq_12} we
  can now estimate the double sum
  \begin{equation*}
    \sum_{n=1}^{N(\hbar)}\sum_{m\geq n}^{N(\hbar)}
    \norm{\chi_{n,\hbar}(\mathcal{H}) \varphi
      \widetilde{\chi}_{m,\hbar}(\mathcal{H})}_{1}\leq \sum_{n=1}^{\infty}\sum_{m\geq n}^{\infty} \frac{C
    }{4^{\frac{3}{4}m - \frac{1}{2}n}} \hbar^{1-d}  \leq \widetilde{C} \hbar^{1-d}.
  \end{equation*}
  The remaining terms in \eqref{A.Full_sum} can be estimated in a
  similar way. The second double sum is estimated by the same argument
  but with the roles of $m$ and $n$ interchanged. To estimate the two
  single sums we only need to introduce one commutator to make the sum
  converge and then use the same arguments as for the double sum. For
  the last term we use a Cauchy-Schwarz inequality.  Adding all our
  estimates up we have the bound
  \begin{equation}\label{A.local_eq_14}
    \norm{g_1(\mathcal{H}) \varphi \boldsymbol{1}_{(0,\varepsilon]}(\mathcal{H})}_{1} \leq C\hbar^{1-d}.
  \end{equation}
  By combining \eqref{A.local_eq_14} with \eqref{A.local_eq_3} and
  \eqref{A.local_eq_13} we get the estimate
  \begin{equation}\label{A.local_eq_19}
    \norm{\boldsymbol{1}_{(-\infty,0]}(\mathcal{H}) \varphi \boldsymbol{1}_{(0,\infty)}(\mathcal{H})}_{1}
    \leq C\hbar^{1-d}.
  \end{equation}
  Since the trace norm satisfies the equality $\norm{A}_1 = \norm{A^{*}}_1$ we also  have the bound,
  \begin{equation}\label{A.local_eq_20}
    \norm{\boldsymbol{1}_{(0,\infty)}(\mathcal{H}) \varphi \boldsymbol{1}_{(-\infty,0]}(\mathcal{H})}_{1}
    \leq C\hbar^{1-d}.
  \end{equation}
  By combining  \eqref{A.local_eq_19} and \eqref{A.local_eq_20} with \eqref{A.local_eq_1} we
  obtain the desired bound for the first part of
  \eqref{A.local_bounds}.

  For the second estimate in \eqref{A.local_bounds} we essentially
  repeat the argument. The main difference occurs when the double
  commutator $[[\varphi,\mathcal{H}],\mathcal{H}]$ is calculated. In
  this case, one has to calculate the commutator
  $[[\varphi Q_i,\mathcal{H}],\mathcal{H}]$. This can be done and one
  obtains the result
  \begin{align*}
    [[\varphi Q_i,&\mathcal{H}],\mathcal{H}] = [[\varphi
    Q_i,H_\hbar],H_\hbar]
    \\
    =& \hbar^2 \sum_{j=1}^d 2 \varphi_{x_j} V_{x_i} Q_j + 2
    \varphi_{x_j} V_{x_j} Q_i - 2 i \hbar \varphi_{x_j} V_{x_j x_i} -
    i \hbar \varphi_{x_j x_j} V_{x_i}
    \\
    &\hbox{} + \hbar^2 \sum_{k=1}^d\Big\{ 2 (\varphi V_{x_i})_{x_k}
    Q_k - i \hbar (\varphi V_{x_i})_{x_k x_k} + \sum_{j=1}^d \big[ - 4
    Q_k \varphi_{x_j x_k} Q_i Q_j
    \\
    &\hbox{} - 4 i \hbar \varphi_{x_j x_k x_k} Q_i Q_j + 2i \hbar
    \varphi_{x_j x_k} Q_i Q_j + 2 \hbar \varphi_{x_j x_j x_k} Q_i Q_k
    + \hbar^2 \varphi_{x_j x_j x_k x_k} Q_i \big] \Big\},
  \end{align*}
  where we have used Assumption~\ref{A.local_assumptions}. From this
  form we can note that again we have a bound of the type
  \begin{equation*}
    \norm{ R_{H_\hbar}(i) [[\varphi Q_i,\mathcal{H}],\mathcal{H}]R_{H_\hbar}(i)  } \leq c\hbar^2,
  \end{equation*}
  since $\mathcal{D}(H_\hbar)\subset \mathcal{D}(Q_j Q_i)$ for all
  $j,i\in\{1,\dots,d\}$.  From here the proof proceeds as above just
  with some extra terms to consider. We omit the details.
\end{proof}

\subsection{Local case without non-critical condition}

In this subsection we will apply the multiscale techniques of
\cite{MR1343781} (see also \cite{MR1631419,MR1240575}). Using this approach
will allow us to remove the non critical assumption on the
potential. Before we state and prove our theorem we will need a lemma
and a remark.
\begin{lemma}\label{A.partition_lemma}
  Let $\Omega\subset\R^d$ be an open set and let $f$ be a function in
  $C^1(\bar{\Omega})$ such that $f>0$ on $\bar{\Omega}$ and assume
  that there exists $\rho$ in $(0,1)$ such that
  \begin{equation}
    \abs{\nabla_x f(x)} \leq \rho,
  \end{equation}
  for all $x$ in $\Omega$. 
  
  Then
  \begin{enumerate}[$i)$]
	
  \item There exists a sequence $\{x_k\}_{k=0}^\infty$ in $\Omega$
    such that the open balls $B(x_k,f(x_k))$ form a covering of
    $\Omega$. Furthermore, there exists a constant $N_\rho$, depending only
    on the constant $\rho$, such that the intersection of more than
    $N_\rho$ balls is empty.
		
  \item One can choose a sequence $\{\varphi_k\}_{k=0}^\infty$ such
    that $\varphi_k \in C_0^\infty(B(x_k,f(x_k)))$ for all $k$ in
    $\N$. Moreover, for all multiindices $\alpha$ and all $k$ in $\N$
    \begin{equation*}
      \abs{\partial_x^\alpha \varphi_k(x)}\leq C_\alpha f(x_k)^{-{\abs{\alpha}}},
    \end{equation*} 	   
    and
    \begin{equation*}
      \sum_{k=1}^\infty \varphi_k(x) = 1,
    \end{equation*}
    for all $x$ in $\Omega$.
		 
  \end{enumerate}
\end{lemma}
This lemma is taken from \cite{MR1343781} where it is Lemma 5.4. The proof is analogous to the proof of \cite[Theorem
1.4.10]{MR1996773}.
\begin{remark}\label{A.con_remark}
  A crucial step in the following proof is scaling of our
  operator. Let $D_f$ and $T_z$, for $f>0$ and $z\in\R^d$, be the unitary dilation and
  translation operators defined by
  \begin{equation*}
    (D_f u)(x) = f^\frac{d}{2} u(fx) ,
  \end{equation*}
  and
  \begin{equation*}
    (T_z u)(x) = u(x+z) ,
  \end{equation*}
  for $u$ in $L^2(\R^d)$. We let $f$ be a positive number and suppose
  $\mathcal{H}$ satisfies Assumption~\ref{A.local_assumptions} with $\Omega$ being
  the open ball $B(z,f)$. We will consider
  the operator
  \begin{equation*}
    \widetilde{\mathcal{H}} = f^{-2} ( T_z U_f) \mathcal{H}( T_z U_f)^{*}.
  \end{equation*}
  The operator $\widetilde{\mathcal{H}}$ is selfadjoint and lower
  semibounded since $\mathcal{H}$ is assumed to be selfadjoint and
  lower semibounded which is the first part of
  Assumption~\ref{A.local_assumptions}. The last part of the
  assumption will be fulfilled with the set $B(0,1)$, the function
  $\widetilde{V}_f(x)=f^{-2} V(fx+z)$ and a scaled $\hbar$ which we will
  call $h$. To see this note that for $\varphi \in C_0^\infty(B(0,1))$
  it holds that $( T_z U_f)^{*} \varphi$ is an element of
  $C_0^\infty(B(z,f))$ since
  \begin{equation*}
    ( T_z U_f)^{*} \varphi(x) = f^{-\frac{d}{2}} \varphi\big( \tfrac{x-z}{f} \big). 
  \end{equation*}
  Hence we have that, using Assumption~\ref{A.local_assumptions} for
  $\mathcal{H}$
  \begin{equation}
    \widetilde{\mathcal{H}} \varphi = -\left(\tfrac{\hbar}{ f^2}\right)^2  \Delta \varphi (x) + f^{-2} V(fx+z) \varphi(x),
  \end{equation}
  This calculation shows that our operator $\widetilde{\mathcal{H}}$
  satisfies Assumption~\ref{A.local_assumptions} with $\Omega=B(0,1)$,
  $V_{loc}=\widetilde{V}_f$ and the new \enquote{Planck's constant}
  $h= \frac{\hbar}{ f^2}$.
\end{remark}
We are now ready to remove the non-critical assumption.
\begin{thm}\label{A.main_local_nc}
  Suppose the operator $\mathcal{H}$ acting in $L^2(\R^d)$ with
  $d\geq2$ obeys Assumption~\ref{A.local_assumptions} with an open set
  $\Omega\subset\R^d$ and let $\hbar_0$ be a strictly positive number.  For $\psi$ in $C_0^\infty(\Omega)$ it holds
  that
  \begin{equation}\label{A.bound_without_nc}
    \norm{[\boldsymbol{1}_{(-\infty,0]}(\mathcal{H}), \psi]}_{1} \leq C\hbar^{1-d}
    \quad\text{and}\quad
    \norm{[\boldsymbol{1}_{(-\infty,0]}(\mathcal{H}),  \psi Q_j]}_{1} \leq  C\hbar^{1-d},
  \end{equation}
  for all $\hbar$ in $(0,\hbar_0]$, where $C$ is a positive constant.
\end{thm}

\begin{proof}
  First note that by assumption $\psi$ is in $C_0^\infty(\Omega)$.
  Hence there exists $\varepsilon>0$ such that
  \begin{equation*}
    \text{dist}(\supp(\psi), \partial\Omega) >\varepsilon.
  \end{equation*}
  We define the function $f$ by
  \begin{equation}\label{A.def_of_func_f}
    f(x) = A^{-1} \left[ V(x)^2 + \abs{\nabla_xV(x)}^4 + \hbar^2 \right]^\frac14 \qquad A>0,
  \end{equation}
  where we have to choose a sufficiently large $A$. It can be noted
  that $f$ is a positive function due to $\hbar$ being a fixed positive
  number. We will need to choose $A$ such that
  \begin{equation}\label{A.l_bound}
    f(x) \leq \frac{\varepsilon}{9} \quad\text{and}\quad \abs{\nabla_x f(x)} \leq \rho < \frac18.
  \end{equation}
  Since $V$ is smooth with compact support $A$ can be chosen such that
  \eqref{A.l_bound} is satisfied. The construction of $f$
  allows us to choose $A$ such that the bounds are valid for all
  $\hbar$ in $(0,\hbar_0]$. Hence $A$ will be independent of $\hbar$, for
  $\hbar$ in the interval $(0,\hbar_0]$. Moreover, we observe that this
  construction gives the estimates
  \begin{equation}\label{A.uniform-est-V_f}
    \abs{V(x)} \leq A f(x)^2, 
    \quad\text{and}\quad
    \abs{\partial_{x_i} V(x)} \leq  A f(x).
  \end{equation}
  This observation will prove useful for controlling bounds on
  some derivatives.

  By Lemma~\ref{A.partition_lemma} with the set $\Omega$ and our
  function $f$ there exists a sequence $\{x_k\}_{k=0}^\infty$ in
  $\Omega$ such that
  $\Omega\subset \bigcup_{k=0}^\infty B(x_k, f(x_k))$ and there exists
  a constant $N_{\frac18}$ in $\mathbb{N}$ such that
  \begin{equation*}
    \bigcap_{k \in \mathcal{I}} B(x_k,f(x_k)) = \emptyset,
  \end{equation*}
  for all $\mathcal{I}\subset \mathbb{N}$ such that
  $\#\mathcal{I}>N_{\frac18}$. Moreover, there exists a sequence
  $\{\varphi_k\}_{k=0}^\infty$ such that
  $\varphi_k \in C_0^\infty(B(x_k,f(x_k)))$,
  \begin{equation*}
    \abs{\partial_x^\alpha \varphi_k} \leq C_\alpha f(x_k)^{-\abs{\alpha}} \qquad \forall \alpha\in \N^d,
  \end{equation*}
  and
  \begin{equation*}
    \sum_{k=1}^\infty \varphi_k(x) = 1 \qquad \forall x\in \Omega.
  \end{equation*}
  Since $\supp(\psi)\subset \Omega$ the union
  $\bigcup_{k=0}^\infty B(x_k, f(x_k))$ forms an open cover of
  $\supp(\psi)$ by assumption the support is compact hence there
  exists $\mathcal{I} \subset \mathbb{N}$ such that $\# \mathcal{I}<\infty$
  and
  \begin{equation*}
    \Omega \subset \bigcup_{k \in \mathcal{I}} B(x_k, f(x_k)).
  \end{equation*}
  We can assume that each ball has a nontrivial intersection with
  $\Omega$. Since at most $N_\frac18$ balls intersect nontrivially we can
  without loss of generality assume that
  \begin{equation*}
    \sum_{k\in \mathcal{I}} \varphi_k(x) = 1 \qquad \forall x\in \supp(\psi).
  \end{equation*}
  From this we get the following estimate:
  \begin{equation}\label{A.Esimate_4.1}
    \norm{[\boldsymbol{1}_{(-\infty,0]}(\mathcal{H}), \psi]}_{1} \leq \sum_{k\in \mathcal{I}}\norm{[\boldsymbol{1}_{(-\infty,0]}(\mathcal{H}),\varphi_k \psi]}_{1}.
  \end{equation}
  We will consider each term separately. We can note that the function
  $\varphi_k\psi$ is smooth and supported in the ball
  $B(x_k,f(x_k))$. The idea is now to make a unitary conjugation of
  our commutator such that a non-critical assumption is obtained.

  Let $T_{x_k}$ be the unitary translation with $x_k$ and let
  $U_{f(x_k)}$ be the unitary scaling operator with $f(x_k)$. We will
  use the notation from Remark~\ref{A.con_remark} and let
  \begin{equation*}
    \widetilde{\varphi_k\psi}(x) = \varphi_k\psi(f(x_k)x+x_k).
  \end{equation*}
  Since the trace norm is invariant under unitary conjugation we have
  that
  \begin{align*}
    \MoveEqLeft \norm{[\boldsymbol{1}_{(-\infty,0]}(\mathcal{H}),\varphi_k
      \psi]}_{1}
    \\
    &=f(x_k)^2 \norm{f(x_k)^{-2} (T_{x_k} U_{f(x_k)})
      [\boldsymbol{1}_{(-\infty,0]}(\mathcal{H}),\varphi_k \psi]
      (T_{x_k} U_{f(x_k)})^*}_{1}
    \\
    &=f(x_k)^2
    \norm{[\boldsymbol{1}_{(-\infty,0]}(\widetilde{\mathcal{H}}),\widetilde{\varphi_k
        \psi}]}_{1}.
  \end{align*}
  By Remark~\ref{A.con_remark}, $\widetilde{\mathcal{H}}$
  satisfies Assumption~\ref{A.local_assumptions} with
  $h=\hbar f(x_k)^{-2}$, $\widetilde{V}_f$ and $B(0,8)$, since by
  construction we have that $B(x_k,8f(x_k)) \subset \Omega$.

  For all $x$ in $B(x_k,8f(x_k))$ we have that
  \begin{equation}\label{A.Main_wonc_pr_1}
  	\begin{aligned}
    \MoveEqLeft f(x) = f(x) - f(x_k) + f(x_k)
    \\
    &\geq - \max_{c\in[0,1]} \abs{\nabla_xf(cx+(1-c)x_k)} \abs{x-x_k}
    +f(x_k)
    \\
    &\geq (1- 8\rho)f(x_k).
    \end{aligned}
  \end{equation}
  Analogously we can note that
  \begin{align}\label{A.Main_wonc_pr_2}
    f(x) \leq (1+ 8\rho)f(x_k),
  \end{align}
  for all $x$ in $B(x_k,8 f(x_k))$. We note that the numbers
  $1\pm 8\rho$ are independent of $k$. The aim is to use
  Theorem~\ref{A.Main_Local_Theorem _wnc}. To see that the non-critical assumption \eqref{A.main_local_nc} is satisfied we note that
  \begin{align*}
    \abs{\widetilde{V}_f(x)} +& h + \abs{\nabla_x \widetilde{V}_f(x)}^2
    \\
    &= f(x_k)^{-2}\left( \abs{V(f(x_k)x+x_k)} + \hbar + \abs{(\nabla_x
        V)(f(x_k) x + x_k )}^2 \right)
    \\
    &= f(x_k)^{-2}\left( \sqrt{\abs{V(f(x_k)x+x_k)}^2} +
      \sqrt{\hbar^2} + \sqrt{\abs{(\nabla_x V)(f(x_k) x + x_k )}^4}
    \right)
    \\
    &\geq f(x_k)^{-2}\left(\abs{V(f(x_k)x+x_k)}^2 +\hbar^2 +
      \abs{(\nabla_x V)(f(x_k) x + x_k )}^4 \right)^{\frac12}
    \\
    &= f(x_k)^{-2} A^2 f(f(x_k) x + x_k )^2
    \\
    &\geq c A^2 >0.
  \end{align*}
Here we used \eqref{A.Main_wonc_pr_1} and \eqref{A.Main_wonc_pr_2} to get the cancelation.
  Therefore the assumption \eqref{A.main_local_nc} is valid for the operator $\widetilde{\mathcal{H}}$. In
  order to ensure uniformity of the error terms from
  Theorem~\ref{A.Main_Local_Theorem _wnc} we need the derivatives of
  $\widetilde{V}_f$ and $\widetilde{\varphi_k \psi}$ to be bounded
  uniformly in $k$. We note that
  \begin{equation*}
    \abs{\partial_x^\alpha \widetilde{V}_f} 
    = \abs{ f(x_k)^{\abs{\alpha}-2} (\partial_x^\alpha V)(f(x_k)x+x_k) } \leq C_\alpha,
  \end{equation*}
  where we in the cases of $\alpha=0$ and $\abs{\alpha}=1$ use the
  estimates from equation \eqref{A.uniform-est-V_f}. For
  $\widetilde{\varphi_k \psi}$ we note that
  \begin{align*}
    \abs{\partial_x^\alpha \widetilde{\varphi_k \psi} } &=
    \abs{f(x_k)^{\abs{\alpha}} \sum_{\beta\leq \alpha}
      {\binom{\alpha}{\beta}} (\partial_x^{\beta}
      \varphi_k)(f(x_k)x+x_k) (\partial_x^{\alpha - \beta}
      \psi)(f(x_k)x+x_k) }
    \\
    &\leq \sum_{\beta\leq \alpha} {\binom{\alpha}{\beta}}
    f(x_k)^{\abs{\alpha-\beta}} \abs{(\partial_x^{\alpha - \beta}
      \psi)(f(x_k)x+x_k) } \leq \widetilde{C}_\alpha.
  \end{align*}
  Lastly we need to verify that the new semiclassical parameter is bounded. By the choice of $A$ we have
  \begin{equation*}
  	h_k= \frac{\hbar}{f(x_k)^2} \leq A^2,
  \end{equation*}
where we have used the definition of the function $f$ \eqref{A.def_of_func_f}. Hence we are in a situation where we can use
  Theorem~\ref{A.Main_Local_Theorem _wnc} which implies that
  \begin{equation}\label{A.Main_wonc_pr_3}
  	\begin{aligned}
    \norm{[\boldsymbol{1}_{(-\infty,0]}(\mathcal{H}),\varphi_k
      \psi]}_{1} &=f(x_k)^2
    \norm{[\boldsymbol{1}_{(-\infty,0]}(\widetilde{\mathcal{H}}),\widetilde{\varphi_k
        \psi}]}_{1}
    \\
    &\leq f(x_k)^2 c \left(\frac{\hbar}{f(x_k)^2}\right)^{1-d}
    \\
    &\leq C \hbar^{1-d} \int_{B(x_k,f(x_k))} f(x)^{d} \,d{x},
    \end{aligned}
  \end{equation}
  with $C$ independent of $k$ in $\mathcal{I}$ and where we also have used \eqref{A.Main_wonc_pr_1} and \eqref{A.Main_wonc_pr_2} in the last estimate. Since $f$ is a bounded function and at most $N_{\frac18}$ of the
  balls $B(x_k,f(x_k))$ can intersect non-empty we get the estimate
  \begin{equation}\label{A.Main_wonc_pr_4}
    \sum_{k\in \mathcal{I}} \int_{B(x_k,f(x_k))} f(x)^{d} \;d{x} \leq
    C(N_{\frac18}) \text{Vol}(\Omega).
  \end{equation}
  By combining \eqref{A.Main_wonc_pr_3} and \eqref{A.Main_wonc_pr_4} with \eqref{A.Esimate_4.1} we
  get the estimate
  \begin{equation*}
    \norm{[\boldsymbol{1}_{(-\infty,0]}(\mathcal{H_\hbar}), \psi]}_{1} \leq \sum_{k\in \mathcal{I}}\norm{[\boldsymbol{1}_{(-\infty,0]}(\mathcal{H_\hbar}),\varphi_k \psi]}_{1} \leq C \hbar^{1-d} ,
  \end{equation*}
  where $C$ depends on the set $\Omega$, the number $N_{\frac18}$, the
  derivatives of $\psi$ and the potential $V$.  We now need to prove
  the second bound in \eqref{A.bound_without_nc}. The proof of this
  bound is completely analogous. Notice that when the
  unitary conjugation is made one should multiply by
  $f(x_k)^3 f(x_k)^{-3}$ instead of $f(x_k)^2 f(x_k)^{-2}$ due to the extra derivative. This ends
  the proof.
\end{proof}

\section{Proof of Theorem~\ref{A.Main_Theorem} and Corollary~\ref{A.cor_Main_Theorem}}\label{A.Proof_of_Main}
In this section we will use the results obtained in the previous
sections to prove Theorem~\ref{A.Main_Theorem} and then use this
theorem to prove Corollary~\ref{A.cor_Main_Theorem}. First the proof
of Theorem~\ref{A.Main_Theorem}:
\begin{proof}[Proof of Theorem~\ref{A.Main_Theorem}]
  Recall that we are in the setting with $H_\hbar=-\hbar^2 \Delta + V$
  being a Schr\"odinger operator acting in $L^2(\R^d)$ with $d\geq2$,
  where $V$ satisfies Assumption~\ref{A.general_assumptions_on_V} and $\hbar$ is bounded by a strictly positive number $\hbar_0$. We
  will here prove the following bounds
  \begin{equation}\label{A.cbound_3}
    \norm{[\boldsymbol{1}_{(-\infty,0]}(H_\hbar),x_i]}_{1} \leq C\hbar^{1-d}
    \quad\text{and}\quad 
    \norm{[\boldsymbol{1}_{(-\infty,0]}(H_\hbar),Q_j]}_{1} \leq C\hbar^{1-d} ,
  \end{equation}
  where $Q_j=-i\hbar\partial_{x_i}$ and $j \in\{1,\dots,d\}$.

  Without loss of generality we can assume that $V$ attains negative
  values. If not, then $H_\hbar$ would be a positive operator with
  purely positive spectrum which implies both commutators would be
  zero and hence satisfy the estimate.

  By assumption we have the open set $\Omega_V$ for which
  $V\in C^\infty(\Omega_V)$ and the bounded set $\Omega_\varepsilon$
  satisfying that $\overline{\Omega}_\varepsilon \subset
  \Omega_V$. Hence we can find an open set $U$  satisfying
  that it is bounded and
  \begin{equation*}
  \Omega_\varepsilon \subset  \subset U \subset \subset \Omega_V
  \end{equation*}
where $ \subset \subset$ means compactly imbedded.
  We let $\chi$ be in $C_0^\infty(U)$ such that $0\leq\chi\leq1$ and
  $\chi(x)=1$ for all $x$ in $\overline{\Omega}_\varepsilon$. Moreover we let
  $\widetilde{\chi}$ be in $C_0^\infty(\Omega_V)$ such that
  $0\leq\widetilde{\chi}\leq1$ and $\widetilde{\chi}(x)=1$ for all $x$ in
  $\overline{U}$. With these sets and functions we have that our
  operator $H_\hbar$ satisfies Assumption~\ref{A.local_assumptions}
  with $\Omega=U$ and $V_{loc}=V\widetilde{\chi}$. With this setup we
  are ready to prove the bounds in \eqref{A.cbound_3}.
 
  We will now consider the first commutator in \eqref{A.cbound_3} and
  note that
  \begin{equation}\label{A.proof_main_1}
    \norm{[\boldsymbol{1}_{(-\infty,0]}(H_\hbar),x_i]}_{1} \leq \norm{[\boldsymbol{1}_{(-\infty,0]}(H_\hbar),\chi x_i]}_{1} 
    + \norm{[\boldsymbol{1}_{(-\infty,0]}(H_\hbar),(1-\chi)x_i]}_{1}.
  \end{equation}
  For the first term in \eqref{A.proof_main_1} we are in a situation where we can use
  Theorem~\ref{A.main_local_nc} since $\chi x_i$ is in
  $C_0^\infty(U)$ and $H_\hbar$ satisfies
  Assumption~\ref{A.local_assumptions} with $\Omega=U$. Then the
  theorem gives us the bound:
  \begin{equation}\label{A.finalbound_1}
    \norm{[\boldsymbol{1}_{(-\infty,0]}(H_\hbar),\chi x_i]}_{1} \leq C\hbar^{1-d}.
  \end{equation}
  For the other term we note that
  \begin{align*}
    \norm{[\boldsymbol{1}_{(-\infty,0]}(H_\hbar),(1-\chi)x_i]}_{1}
    &\leq \norm{\boldsymbol{1}_{(-\infty,0]}(H_\hbar)(1-\chi)x_i}_{1}
    +\norm{(1-\chi)x_i\boldsymbol{1}_{(-\infty,0]}(H_\hbar)}_{1}
    \\
    &= 2\norm{\boldsymbol{1}_{(-\infty,0]}(H_\hbar)(1-\chi)x_i}_{1}.
  \end{align*}
 By a Cauchy-Schwarz inequality we have that
  \begin{align}\label{A.finalbound_2}
    \begin{split}
      \norm{\boldsymbol{1}_{(-\infty,0]}(H_\hbar)(1-\chi)x_i}_{1}
      &\leq \norm{\boldsymbol{1}_{(-\infty,0]}(H_\hbar)}_{2}
      \norm{\boldsymbol{1}_{(-\infty,0]}(H_\hbar)(1-\chi)x_i}_{2}
      \\
      &= \Tr(\boldsymbol{1}_{(-\infty,0]}(H_\hbar))^\frac12
      \norm{\boldsymbol{1}_{(-\infty,0]}(H_\hbar)(1-\chi)x_i}_{2}.
    \end{split}
  \end{align}
  The first term squared can be estimated by a constant times
  $\hbar^{-\frac{d}{2}}$ by Remark~\ref{A.clr_bound_remark}. For the
  second term we calculate the trace in a basis of eigenfunctions for
  $H_\hbar$.
  \begin{align}\label{A.finalbound_3}
    \begin{split}
      \norm{\boldsymbol{1}_{(-\infty,0]}(H_\hbar)(1-\chi)x_i}_{2}^2 &=
      \Tr[\boldsymbol{1}_{(-\infty,0]}(H_\hbar)(1-\chi)x_i^2(1-\chi)\boldsymbol{1}_{(-\infty,0]}(H_\hbar)]
      \\
      &= \sum_{\lambda_n \leq \frac{\varepsilon}{4}}
      \scp{\boldsymbol{1}_{(-\infty,0]}(H_\hbar)(1-\chi)x_i^2(1-\chi)\boldsymbol{1}_{(-\infty,0]}(H_\hbar)\psi_n}{\psi_n}
      \\
      &= \sum_{\lambda_n \leq 0} \norm{(1-\chi)x_i\psi_n}_{L^2(\R^d)}^2.
    \end{split}
  \end{align}
  In order to estimate the $L^2(\R^d)$-norm, we let
  $d(x)=\mathrm{dist}(x,\Omega_\varepsilon)$. For all $x$ in the
  support of $1-\chi$ we have that $d(x)>0$ since
  $\Omega_\varepsilon$ is a proper subset of the support of $\chi$. We
  can note that $V$ is an element of $L_{loc}^1(\R^d)$ hence
  Lemma~\ref{A.Agmon_type_lem} gives the existence of a constant $C$
  only depending on $V$ such that for all eigenvectors $\psi_n$ with
  eigenvalue less than $\tfrac{\varepsilon}{4}$ we have the estimate
  \begin{align*}
    \norm{e^{\delta d \hbar^{-1}} \psi_n }_{L^2(\R^d)} \leq C,
  \end{align*}
  where $\delta=\tfrac{\sqrt{\varepsilon}}{8}$. With these
  observations we can note that for all norms in the last sum of
  \eqref{A.finalbound_3} we have for all $N$ in $\N$ the bound
  \begin{equation}\label{A.trick_small_norm_end}
  	\begin{aligned}
  \norm{(1-\chi)x_i\psi_n}_{L^2(\R^d)}^2 &\leq \norm{(1-\chi)x_i e^{-\delta\varphi\hbar^{-1}}}_\infty^2
    \norm{e^{\delta\varphi\hbar^{-1}} \psi_n}_{L^2(\R^d)}^2
    \\
    &\leq C \norm{(1-\chi)x_i (\frac {\hbar}{\delta \varphi})^N
      (\frac{\delta \varphi}{\hbar})^N
      e^{-\delta\varphi\hbar^{-1}}}_\infty^2
    \\
    &\leq C_N \hbar^{2N},
    \end{aligned}
  \end{equation}
  where the constant depends on the choice of the set $U$,
  $\delta(\varepsilon)$ and the power $N$. If we now combine this
  estimate with \eqref{A.finalbound_3} we get
  \begin{equation}\label{A.final_bound_4}
    \norm{\boldsymbol{1}_{(-\infty,0]}(H_\hbar)(1-\chi)x_i}_{2}^2 \leq C \hbar^{2N-d} ,
  \end{equation}
  where we have used Remark~\ref{A.clr_bound_remark} to estimate the
  number of terms in the sum in \eqref{A.finalbound_3}. Combining \eqref{A.final_bound_4}
  with \eqref{A.finalbound_2} we get
  \begin{equation*}
    \norm{\boldsymbol{1}_{(-\infty,0]}(H_\hbar)(1-\chi)x_i}_{1} \leq C_N \hbar^{N-d}.
  \end{equation*}
  Now by combining this bound with \eqref{A.finalbound_1} we get the
  desired bound in \eqref{A.cbound_3}.

  For the second bound in \eqref{A.cbound_3} we take the same $\chi$
  as above and note that
  \begin{equation*}
    \norm{[\boldsymbol{1}_{(-\infty,0]}(H_\hbar),Q_i]}_{1} \leq \norm{[\boldsymbol{1}_{(-\infty,0]}(H_\hbar),\chi Q_i]}_{1} + \norm{[\boldsymbol{1}_{(-\infty,0]}(H_\hbar),(1-\chi)Q_i]}_{1}.
  \end{equation*}
  The first term can as above be estimated by applying
  Theorem~\ref{A.main_local_nc}. The second term will be proven to be
  small as before. We note that
  \begin{align*}
    \norm{[\boldsymbol{1}_{(-\infty,0]}(H_\hbar),(1-\chi)Q_i]}_{1}
    &\leq \norm{\boldsymbol{1}_{(-\infty,0]}(H_\hbar)(1-\chi)Q_i}_{1}
    + \norm{(1-\chi)Q_i\boldsymbol{1}_{(-\infty,0]}(H_\hbar)}_{1}
    \\
    &\leq 2
    \norm{\boldsymbol{1}_{(-\infty,0]}(H_\hbar)(1-\chi)Q_i}_{1} +\hbar
    \norm{\boldsymbol{1}_{(-\infty,0]}(H_\hbar)\partial_{x_i}\chi}_{1}.
  \end{align*}
  The second term is on the same form as the left hand side of \eqref{A.finalbound_2} and hence can be treated as above. For the first term we
  have that
  \begin{equation*}
    \norm{\boldsymbol{1}_{(-\infty,0]}(H_\hbar)(1-\chi)Q_i}_{1}  
    \leq
    \norm{\boldsymbol{1}_{(-\infty,0]}(H_\hbar)}_{2}\norm{\boldsymbol{1}_{(-\infty,0]}(H_\hbar)(1-\chi)Q_i}_{2}.
  \end{equation*}
  The first term can be controlled by
  Remark~\ref{A.clr_bound_remark}. For the second term we have that
  \begin{align*}
    \norm{\boldsymbol{1}_{(-\infty,0]}(H_\hbar)(1-\chi)Q_i}_{2} =&
    \norm{\boldsymbol{1}_{(-\infty,0]}(H_\hbar)(1-\chi)Q_i^2(1-\chi)\boldsymbol{1}_{(-\infty,0]}(H_\hbar)}_{1}^{\frac{1}{2}}
    \\
    \leq&
    \norm{\boldsymbol{1}_{(-\infty,0]}(H_\hbar)(1-\chi)(H_\hbar+c)(1-\chi)\boldsymbol{1}_{(-\infty,0]}(H_\hbar)}_{1}^{\frac{1}{2}},
  \end{align*}
  where 
  \begin{equation}
  c=1-\inf_{x \in\Omega_\varepsilon}(V(x)).
  \end{equation}
  If we now calculate the trace norm by choosing a basis of
  eigenfunctions of $H_\hbar$ we get that
  \begin{align*}
    \lVert\boldsymbol{1}_{(-\infty,0]}(H_\hbar)&(1-\chi)(H_\hbar+c)(1-\chi)\boldsymbol{1}_{(-\infty,0]}(H_\hbar)\rVert_{1}
    \\
    &= \sum_{\lambda_n \leq \frac{\varepsilon}{4}}
    \scp{\boldsymbol{1}_{(-\infty,0]}(H_\hbar)(1-\chi)(H_\hbar+c)(1-\chi)\boldsymbol{1}_{(-\infty,0]}(H_\hbar)\psi_n}{\psi_n}.
  \end{align*}
  If we consider just one of the terms we have by the IMS formula
  that
  \begin{align*}
   \MoveEqLeft \scp{\boldsymbol{1}_{(-\infty,0]}(H_\hbar)(1-\chi)(H_\hbar+c)(1-\chi)\boldsymbol{1}_{(-\infty,0]}(H_\hbar)\psi_n}{\psi_n}
    \\
    ={}& c \scp{(1-\chi)\boldsymbol{1}_{(-\infty,0]}(H_\hbar)\psi_n}{(1-\chi)\boldsymbol{1}_{(-\infty,0]}(H_\hbar)\psi_n} 
	\\	
	&+ \scp{H_\hbar(1-\chi)\boldsymbol{1}_{(-\infty,0]}(H_\hbar)\psi_n}{(1-\chi)\boldsymbol{1}_{(-\infty,0]}(H_\hbar)\psi_n}
    \\
    ={}& c \scp{(1-\chi)\boldsymbol{1}_{(-\infty,0]}(H_\hbar)\psi_n}{(1-\chi)\boldsymbol{1}_{(-\infty,0]}(H_\hbar)\psi_n}
    \\
    &+\scp{(1-\chi)H_\hbar \boldsymbol{1}_{(-\infty,0]}(H_\hbar)\psi_n}{(1-\chi)\boldsymbol{1}_{(-\infty,0]}(H_\hbar)\psi_n}
    \\
    &+ \hbar^2 \int_{\R^d} \abs{\nabla_x\chi}^2 \abs{\psi_n}^2
    \; d{x}
    \\
    \leq{}& (c+\lambda_n)\norm{(1-\chi)\psi_n}_{L^2(\R^d)}^2 + \hbar^2
    \norm{ \abs{\nabla_x\chi} \psi_n}_{L^2(\R^d)}^2.
  \end{align*}
  We can note that the number $ c+\lambda_n$ is less than or equal to
  $c+\frac{\varepsilon}{2}$ for the possible values of
  $\lambda_n$. For the two norms we can use the same trick as in \eqref{A.trick_small_norm_end} and thereby show that they are small in $\hbar$. This completes the proof.
\end{proof}
Now the proof of the corollary:
\begin{proof}[Proof of Corollary~\ref{A.cor_Main_Theorem}]
  We start by observing that the operator
  \begin{equation*}
  	[\boldsymbol{1}_{(-\infty,0]}(H_\hbar),x],
\end{equation*}
 is a trace class
  operator by Theorem~\ref{A.Main_Theorem}, where the commutator is
  interpreted as the sum of the commutators with each entry in the
  vector $x$. Moreover we note that
  \begin{align}\label{A.exp_com}
    \begin{split}
      [\boldsymbol{1}_{(-\infty,0]}(H_\hbar),e^{i \langle t,x\rangle}]
      &= \boldsymbol{1}_{(-\infty,0]}(H_\hbar)e^{i \langle t,x\rangle}
      - e^{i \langle t,x\rangle} \boldsymbol{1}_{(-\infty,0]}(H_\hbar)
      \\
      &= e^{i \langle t,x\rangle}\Big( e^{-i \langle
        t,x\rangle}\boldsymbol{1}_{(-\infty,0]}(H_\hbar)e^{i \langle
        t,x\rangle} - \boldsymbol{1}_{(-\infty,0]}(H_\hbar)\Big).
    \end{split}
  \end{align}
  We define the function $f:\R\rightarrow\mathcal{B}(L^2(\R^d))$,
  where $\mathcal{B}(L^2(\R^d))$ are the bounded operators on
  $L^2(\R^d)$, by
  \begin{equation*}
    f(s) = e^{-i \langle t,x\rangle s}\boldsymbol{1}_{(-\infty,0]}(H_\hbar)e^{i\langle t,x\rangle s}. 
  \end{equation*}
  For this function we note that
  \begin{equation*}
    e^{i \langle t,x\rangle}(f(1) - f(0)) = [\boldsymbol{1}_{(-\infty,0]}(H_\hbar),e^{i \langle t,x\rangle}].
  \end{equation*}
  By \eqref{A.exp_com} we have that
  \begin{align*}
    \frac{d}{ds}f(s) &= - i \langle t,x\rangle e^{-i\langle t,x\rangle
      s} \boldsymbol{1}_{(-\infty,0]}(H_\hbar) e^{i\langle t,x\rangle
      s} + i e^{-\langle t,x\rangle s}
    \boldsymbol{1}_{(-\infty,0]}(H_\hbar) \langle t,x\rangle
    e^{i\langle t,x\rangle s}
    \\
    &=i e^{- i \langle t,x\rangle s}
    [\boldsymbol{1}_{(-\infty,0]}(H_\hbar) ,\langle t,x\rangle ]
    e^{i\langle t,x\rangle s}.
  \end{align*}
  With this we note by the fundamental theorem of calculus that
  \begin{align*}
    \norm{ [\boldsymbol{1}_{(-\infty,0]}(H_\hbar),e^{i \langle
        t,x\rangle}]}_1 &= \norm{\int_0^1 e^{i\langle t,x\rangle(1-
        s)} [\boldsymbol{1}_{(-\infty,0]}(H_\hbar) ,\langle t,x\rangle
      ] e^{i\langle t,x\rangle s} \; ds}_1
    \\
    &\leq \norm{ [\boldsymbol{1}_{(-\infty,0]}(H_\hbar) ,\langle
      t,x\rangle ]}_1 \leq \sum_{j=1}^d \abs{t_j} \norm{
      [\boldsymbol{1}_{(-\infty,0]}(H_\hbar) ,x_j ]}_1.
  \end{align*}
  With this bound the desired result follows from
  Theorem~\ref{A.Main_Theorem}.
\end{proof}

\appendix
\section{Agmon type estimates}

In this appendix we will prove an Agmon type estimate, that is exponential decay of eigenfunctions for a Schr\"{o}dinger operator. Such results were proven by S. Agmon see \cite{MR745286}.

\begin{lemma}\label{A.Agmon_type_lem}
  Let $H_\hbar=-\hbar^2\Delta +V$ be a Schr\"odinger operator acting
  in $L^2(\R^d)$, where $V$ is in $L^1_{loc}(\R^d)$ and suppose that
  there exist an $\varepsilon>0$ and a open bounded sets $U$ such that
  \begin{equation*}
    V(x) \geq {\varepsilon} \quad \text{when } x\in U^c.
  \end{equation*}
  Let $d(x) = \dist(x,\Omega_\varepsilon)$ and $\psi$ be a normalised
  solution to the equation
  \begin{equation*}
    H_\hbar\psi=E\psi,
  \end{equation*}
  with $E<\varepsilon/4$. Then there exists a $C>0$ depending on $V$ and $\varepsilon$ such that
  \begin{equation*}
    \norm{e^{\delta \hbar^{-1} d } \psi }_{L^2(\R^d)} \leq C,
  \end{equation*}
for $\delta=\tfrac{\sqrt{\varepsilon}}{8}$.
\end{lemma}

\begin{proof}
  We start by defining the set $\Omega_\varepsilon$ by
  \begin{equation*}
    \Omega_\varepsilon = \{ x\in\R^d \, |\, \dist(x,U)<1\}.
  \end{equation*}
  For convenience and without loss of generality we assume that
  $0\in U$, which implies that $d(x)\leq\abs{x}$ for all
  $x$ in $\R^d$. For $\gamma\in(0,1]$ we define the function
  $\varphi_\gamma$ by
  \begin{equation*}
    \varphi_\gamma(x) = \frac{d(x)}{1+\gamma \abs{x}^2}.
  \end{equation*}
  Then $\varphi_\gamma$ is a bounded function for all $\gamma$'s by
  construction. Moreover we can note that $d(x)$ is almost everywhere
  differentiable with the norm of the gradient bounded by 1 since it
  is Lipschitz continuous with Lipschitz constant 1. Hence
  $ \varphi_\gamma$ is almost everywhere differentiable. We will prove
  the bound on the 2-norm is uniform in the parameter $\gamma$ for the
  functions $\varphi_\gamma$ and let $\gamma$ tend to zero.
	
  In order to prove the desired bound we need a partition of unity. We
  let $\chi:\R^d\rightarrow\R$ be a smooth function such that
  $0\leq\chi\leq1$, $\chi(x)=1$ for all $x$ in $\Omega_\varepsilon^c$
  and $\mathrm{Supp}(\chi)\subset U^c$. For this function we note that
  \begin{align*}
    \norm{e^{\delta\varphi_\gamma\hbar^{-1}} \psi }_{L^2(\R^d)} \leq&
    \norm{e^{\delta\varphi_\gamma\hbar^{-1}}(1-\chi) \psi
    }_{L^2(\R^d)} + \norm{e^{\delta\varphi_\gamma\hbar^{-1}}\chi\psi
    }_{L^2(\R^d)}
    \\
    \leq& 1 + \norm{e^{\delta\varphi_\gamma\hbar^{-1}}\chi\psi
    }_{L^2(\R^d)},
  \end{align*}
  where we have used that $1-\chi$ is supported in
  $\Omega_\varepsilon$ and $\varphi_\gamma(x)=0$ for
  $x\in\Omega_\varepsilon$. Since $\varphi_\gamma$ is a bounded
  function the left hand side in the above inequality is well defined. What
  remains is to estimate the last term in the above inequality.

  To this end we note that since $\psi$ is an eigenfunction with
  eigenvalue $E$ we have that
  \begin{align*}
    (\tfrac{\varepsilon}{2}-E)\norm{e^{\delta\varphi_\gamma\hbar^{-1}}\chi\psi
    }_{L^2(\R^d)}^2 &= (\tfrac{\varepsilon}{2}-E)\int_{\R^d}
    e^{2\delta\varphi_\gamma\hbar^{-1}}\chi^2\abs{\psi}^2 \;dx
    \\
    &= \langle e^{2\delta\varphi_\gamma\hbar^{-1}}\chi^2 \psi,
    (\tfrac{\varepsilon}{2}-H) \psi \rangle.
  \end{align*}
  Note that the above expression is real, hence we can take the real
  part of the right hand side without changing it. If we do this and
  use the IMS-formula we get that
  \begin{align*}
    \Real(\langle e^{2\delta\varphi_\gamma\hbar^{-1}}\chi^2 \psi,
    (\tfrac{\varepsilon}{2}-H) \psi \rangle)    =& \Real(\langle e^{\delta\varphi_\gamma\hbar^{-1}}\chi \psi ,
    (\tfrac{\varepsilon}{2}-H) e^{\delta\varphi_\gamma\hbar^{-1}}\chi
    \psi \rangle)
    \\
    &\hbox{}+ \hbar^2 \int_{\R^d} \abs{\nabla
      e^{\delta\varphi_\gamma\hbar^{-1}}\chi}^2 \abs{\psi}^2 \; d{x}.
  \end{align*}
  Note that the above gradient is well defined almost everywhere due
  to our previous observations. Since
  $e^{\delta\varphi_\gamma\hbar^{-1}}\chi \psi \in \mathcal{Q}(H)$
  and is supported in $U^c$ we have that
  \begin{equation*}
    \Real(\langle e^{\delta\varphi_\gamma\hbar^{-1}}\chi \psi 
    , (\tfrac{\varepsilon}{2}-H) e^{\delta\varphi_\gamma\hbar^{-1}}\chi \psi  \rangle) \leq 0,
  \end{equation*}
  since $(\tfrac{\varepsilon}{2}-H)$ is a negative operator when
  restricted to $U^c$. From this we obtain the inequality
  \begin{equation*}
    (\tfrac{\varepsilon}{2}-E)\norm{e^{\delta\varphi_\gamma\hbar^{-1}}\chi\psi }_{L^2(\R^d)}^2  
    \leq 
    \hbar^2 \int_{\R^d} \abs{\nabla e^{\delta\varphi_\gamma\hbar^{-1}}\chi}^2\abs{\psi}^2 \; d{x}.
  \end{equation*}
  We note that
  \begin{equation}\label{A.agmont_bd_gra}
    \abs{\nabla e^{\delta\varphi_\gamma\hbar^{-1}}\chi}^2
    \leq
    4\abs{\nabla e^{\delta\varphi_\gamma\hbar^{-1}}}^2 \chi^2 
    +4 e^{2\delta\varphi_\gamma\hbar^{-1}} \abs{\nabla  \chi}^2,
  \end{equation}
  where the gradients are defined almost everywhere with respect to the
  Lebesgue measure. The first term in \eqref{A.agmont_bd_gra} is
  almost everywhere given by
  \begin{align*}
    4\abs{\nabla e^{\delta\varphi_\gamma\hbar^{-1}}}^2 \chi^2 =
    4\frac{\delta^2}{\hbar^2} \abs{\nabla \varphi_\gamma}^2
    e^{2\delta\varphi_\gamma\hbar^{-1}}\chi^2.
  \end{align*}
  We note that for $x$ in $\Omega_\varepsilon$
  $\abs{\nabla \varphi_\gamma(x)}=0$, and for almost all $x$ in
  $\Omega_\varepsilon^c$ 
  \begin{equation*}
    \abs{\nabla \varphi_\gamma(x)} 
    \leq \frac{\abs{\nabla d(x)}}{1+\gamma \abs{x}^2} + 2 \frac{d(x)\gamma\abs{x}}{(1+\gamma\abs{x}^2)^2} 
    \leq 1 +2 \frac{\gamma\abs{x}^2}{(1+\gamma\abs{x}^2)^2} \leq 2.
  \end{equation*}
  Hence for all $x$ in $\R^d$ we have,
  \begin{equation*}
    \abs{\nabla \varphi_\gamma(x)} \leq  2.
  \end{equation*}
  With these estimates we get that
  \begin{align*}
    (\tfrac{\varepsilon}{2}-E)&\norm{e^{\delta\varphi_\gamma\hbar^{-1}}\chi\psi
    }_{L^2(\R^d)}^2
    \\
    &\leq 8 \delta^2 \int_{\R^d}
    e^{2\delta\varphi_\gamma\hbar^{-1}}\chi^2 \abs{\psi}^2 \; d{x} +
    4\int_{\R^d} e^{2\delta\varphi_\gamma\hbar^{-1}} \abs{\nabla
      \chi}^2 \abs{\psi}^2 \; d{x}
    \\
    &= 8 \delta^2 \norm{e^{\delta\varphi_\gamma\hbar^{-1}}\chi\psi
    }_{L^2(\R^d)}^2 + 4\int_{\R^d} e^{2\delta\varphi_\gamma\hbar^{-1}}
    \abs{\nabla \chi}^2 \abs{\psi}^2 \; d{x}.
  \end{align*}
  This implies that
  \begin{equation*}
    (\tfrac{\varepsilon}{2}-E-  8 \delta^2)\norm{e^{\delta\varphi_\gamma\hbar^{-1}}\chi\psi }_{L^2(\R^d)}^2  
    \leq
    4\int_{\R^d} e^{2\delta\varphi_\gamma\hbar^{-1}} \abs{\nabla  \chi}^2 \abs{\psi}^2 \; d{x}	.
  \end{equation*}	
  With our choice of $\delta=\frac{\sqrt{\varepsilon}}{8}$ we have
  that
  \begin{equation*}
    (\tfrac{\varepsilon}{2}-E-  8 \delta^2)\geq \frac{\varepsilon}{2} - \frac{\varepsilon}{4} - 8\frac{\varepsilon}{64} 
    = \frac{\varepsilon}{8},
  \end{equation*}	
  which implies that
  \begin{equation*}
    \norm{e^{\delta\varphi_\gamma\hbar^{-1}}\chi\psi }_{L^2(\R^d)}^2  
    \leq
    \frac{32}{\varepsilon} \int_{\R^d} e^{2\delta\varphi_\gamma\hbar^{-1}} \abs{\nabla  \chi}^2 
    \abs{\psi}^2 \; d{x}	.
  \end{equation*}			
  We note that $\abs{\nabla \chi}^2 $ is supported on the set
  $\Omega_\varepsilon\setminus U$ and hence uniformly bounded by a
  constant which depends on the sets. Hence we get that
  \begin{align*}
    \int_{\R^d} e^{2\delta\varphi_\gamma\hbar^{-1}} \abs{\nabla
      \chi}^2 \abs{\psi}^2 \; d{x} &\leq C
    \int_{\Omega_\varepsilon\setminus U}
    e^{2\delta\varphi_\gamma\hbar^{-1}} \abs{\psi}^2 \; d{x}
    \\
    &\leq C \int_{\Omega_\varepsilon\setminus U} \abs{\psi}^2 \; d{x}
    \leq C,
  \end{align*}
  where we have used that $e^{2\delta\varphi_\gamma\hbar^{-1}}=1$ for
  all $x$ in $\Omega_\varepsilon$. This implies that there exists a
  constant $C>0$ which only depends on the potential $V$ such that
  \begin{equation*}
    \norm{e^{\delta\varphi_\gamma\hbar^{-1}}\chi\psi }_{L^2(\R^d)}^2 \leq C .	
  \end{equation*}	
  This estimate implies that we have the following uniform bound in
  $\gamma$
  \begin{equation*}
    \norm{e^{\delta\varphi_\gamma\hbar^{-1}} \psi }_{L^2(\R^d)} 
    \leq  
    1 + C.
  \end{equation*}
  By monotone convergence we can take $\gamma$ to zero and we obtain
  the desired result:
  \begin{equation*}
    \norm{e^{\delta\varphi\hbar^{-1}} \psi }_{L^2(\R^d)}
    \leq  
    C,
  \end{equation*}
  with a constant only depending on the potential $V$.
\end{proof}

\bibliographystyle{plain} \bibliography{Bib_paperA.bib}

\begin{thebibliography}{10}

\bibitem{MR745286}
S.~Agmon.
\newblock {\em Lectures on exponential decay of solutions of second-order
  elliptic equations: bounds on eigenfunctions of {$N$}-body {S}chr\"{o}dinger
  operators}, volume~29 of {\em Mathematical Notes}.
\newblock Princeton University Press, Princeton, NJ; University of Tokyo Press,
  Tokyo, 1982.

\bibitem{MR3570479}
N.~Benedikter, V.~Jak\v{s}i\'{c}, M.~Porta, C.~Saffirio, and B.~Schlein.
\newblock Mean-field evolution of fermionic mixed states.
\newblock {\em Comm. Pure Appl. Math.}, 69(12):2250--2303, 2016.

\bibitem{MR3202863}
N.~Benedikter, M.~Porta, and B.~Schlein.
\newblock Mean-field dynamics of fermions with relativistic dispersion.
\newblock {\em J. Math. Phys.}, 55(2):021901, 10, 2014.

\bibitem{MR3248060}
N.~Benedikter, M.~Porta, and B.~Schlein.
\newblock Mean-field evolution of fermionic systems.
\newblock {\em Comm. Math. Phys.}, 331(3):1087--1131, 2014.

\bibitem{MR3381147}
N.~Benedikter, M.~Porta, and B.~Schlein.
\newblock Hartree-{F}ock dynamics for weakly interacting fermions.
\newblock In {\em Mathematical results in quantum mechanics}, pages 177--189.
  World Sci. Publ., Hackensack, NJ, 2015.

\bibitem{MR1349825}
E.~B. Davies.
\newblock {\em Spectral theory and differential operators}, volume~42 of {\em
  Cambridge Studies in Advanced Mathematics}.
\newblock Cambridge University Press, Cambridge, 1995.

\bibitem{MR1735654}
M.~Dimassi and J.~Sj\"{o}strand.
\newblock {\em Spectral asymptotics in the semi-classical limit}, volume 268 of
  {\em London Mathematical Society Lecture Note Series}.
\newblock Cambridge University Press, Cambridge, 1999.

\bibitem{MR1996773}
L.~H\"{o}rmander.
\newblock {\em The analysis of linear partial differential operators. {I}}.
\newblock Classics in Mathematics. Springer-Verlag, Berlin, 2003.
\newblock Distribution theory and Fourier analysis, Reprint of the second
  (1990) edition [Springer, Berlin; MR1065993 (91m:35001a)].

\bibitem{MR1631419}
V.~Ivrii.
\newblock {\em Microlocal Analysis and Precise Spectral Asymptotics}.
\newblock Springer Monographs in Mathematics. Springer-Verlag, Berlin, 1998.

\bibitem{MR1240575}
V.~Ja. Ivrii and I.~M. Sigal.
\newblock Asymptotics of the ground state energies of large {C}oulomb systems.
\newblock {\em Ann. of Math. (2)}, 138(2):243--335, 1993.

\bibitem{MR4009687}
N.~Leopold and S.~Petrat.
\newblock Mean-field dynamics for the {N}elson model with fermions.
\newblock {\em Ann. Henri Poincar\'{e}}, 20(10):3471--3508, 2019.

\bibitem{MR2583992}
E.~H. Lieb and R.~Seiringer.
\newblock {\em The stability of matter in quantum mechanics}.
\newblock Cambridge University Press, Cambridge, 2010.

\bibitem{MR0493420}
M.~Reed and B.~Simon.
\newblock {\em Methods of modern mathematical physics. {II}. {F}ourier
  analysis, self-adjointness}.
\newblock Academic Press [Harcourt Brace Jovanovich, Publishers], New
  York-London, 1975.

\bibitem{MR897108}
D.~Robert.
\newblock {\em Autour de l'approximation semi-classique}, volume~68 of {\em
  Progress in Mathematics}.
\newblock Birkh\"{a}user Boston, Inc., Boston, MA, 1987.

\bibitem{MR1343781}
A.~V. Sobolev.
\newblock Quasi-classical asymptotics of local {R}iesz means for the
  {S}chr\"{o}dinger operator in a moderate magnetic field.
\newblock {\em Ann. Inst. H. Poincar\'{e} Phys. Th\'{e}or.}, 62(4):325--360,
  1995.

\bibitem{MR2952218}
M.~Zworski.
\newblock {\em Semiclassical analysis}, volume 138 of {\em Graduate Studies in
  Mathematics}.
\newblock American Mathematical Society, Providence, RI, 2012.

\end{thebibliography}

\end{document}